\documentclass[a4paper,UKenglish,cleveref, autoref, thm-restate]{lipics-v2021}

\title{The Calissons Puzzle} 

\author{Favreau Jean-Marie}{Université Clermont Auvergne, LIMOS, France \and \url{http://www.myhomepage.edu} }{j-marie.favreau@uca.fr}{https://orcid.org/0000-0002-2460-6336}{French ANR PRC grant (ANR-19-CE19-0005)}

\author{Gerard Yan
}{Université Clermont Auvergne, LIMOS, France}{yan.gerard@uca.fr}{[https://orcid.org/0000-0002-2664-0650]}{French ANR PRC grant ADDS (ANR-19-CE48-0005)}

\author{Lafourcade Pascal}{Université Clermont Auvergne, LIMOS, France}{pascal.lafourcade@uca.fr}{[https://orcid.org/0000-0002-4459-511X]}{ANR PRC grant MobiS5 (ANR-18-CE39-0019), DECRYPT (ANR-18-CE39-0007), SEVERITAS (ANR-20-CE39-0005) and by the French government IDEX-ISITE initiative 16-IDEX-0001 (CAP 20-25)}

\author{Robert Léo}{Université Clermont Auvergne, LIMOS, France}{leo.robert@uca.fr}{[https://orcid.org/0000-0002-9638-3143]}{ANR PRC grant MobiS5 (ANR-18-CE39-0019)}

\authorrunning{J.M. Favreau, Y. Gerard, P. Lafourcade, L. Robert. } 

\Copyright{Jean-Marie Favreau, Yan Gerard, Pascal Lafourcade, Léo Robert} 



\category{} 

\relatedversion{} 

\acknowledgements{The authors would like to thank Guilherme Da Fonseca for discussing the questions and results of the paper.}

\nolinenumbers 
\EventEditors{John Q. Open and Joan R. Access}
\EventNoEds{2}
\EventLongTitle{42nd Conference on Very Important Topics (CVIT 2016)}
\EventShortTitle{CVIT 2016}
\EventAcronym{CVIT}
\EventYear{2016}
\EventDate{December 24--27, 2016}
\EventLocation{Little Whinging, United Kingdom}
\EventLogo{}
\SeriesVolume{42}
\ArticleNo{23}
{\bfseries}{\itshape}
\newcommand{\etal}

\newcommand{\R}{\ensuremath{\mathbb{R}}}
\newcommand{\Z}{\ensuremath{\mathbb{Z}}}

\def\back{\mathrm{Back}}
\def\front{\mathrm{Front}}
\def\H{\mathcal{H}}

\hideLIPIcs

\title{The Calissons Puzzle}

\newlength{\strutdepth}%
\newcommand{\mycomment}[3]{%
    \noindent{\bfseries
    \color{#2}{#1}\color{black}}%
    \strut\vadjust{\kern-\strutdepth%
        \vtop to \strutdepth{%
            \baselineskip\strutdepth%
            \vss\llap{{\large\color{#2}#3\quad\color{black}}}\null%
        }%
    }%
}

\usepackage{amsmath, amssymb}
\usepackage{mathtools}
\usepackage{dashbox}
\usepackage{bm}
\usepackage{hyperref}
\usepackage{paralist}
\usepackage{graphicx}
\usepackage{cite}
\usepackage{amssymb}
\usepackage{setspace}
\usepackage{colortbl}
\usepackage[table]{xcolor}
\usepackage{textcomp}
\usepackage{tabularx}
\usepackage{adjustbox}
\usepackage{bm}
\usepackage{thmtools} 

\usepackage{tikz}
\usetikzlibrary{fit}
\usetikzlibrary{shapes}
\usetikzlibrary{decorations.markings}
\usetikzlibrary{arrows}
\usetikzlibrary{shapes.geometric}
\usepackage{tkz-base,tkz-fct}
\usepackage{pgfplots}
\usepackage{pgfplotstable}
\pgfplotsset{compat=1.18}

\usepackage{pifont}

\usepackage{enumitem}

\usepackage[strings]{underscore}
\usepackage{csquotes} 
\begin{document}
\maketitle
\begin{abstract}
In 2022, Olivier Longuet, a French mathematics teacher, created a game called the \textit{calissons puzzle}. Given a triangular grid in a hexagon and some given edges of the grid, the problem is to find a calisson tiling such that no input edge is overlapped and calissons adjacent to an input edge have different orientations. 
We extend the puzzle to regions  $R$ that are not necessarily hexagonal. 
The first interesting property of this puzzle is that, unlike the usual calisson or domino problems, it is solved neither by a maximal matching algorithm, nor by Thurston's algorithm. This raises the question of its complexity.

We prove that if the region $R$ is finite and simply connected, then the puzzle can be solved by an algorithm that we call the \textit{advancing surface algorithm} and whose complexity is $O(|\partial R|^3)$ where $\partial R|$ is the size of the boundary of the region $R$. In the case where the region is the entire infinite triangular grid, we prove that the existence of a solution can be solved with an algorithm of complexity $O(|X|^3)$ where $X$ is the set of input edges. To prove these theorems, we revisit William Thurston's results on the calisson tilability of a region $R$. The solutions involve equivalence between calisson tilings, stepped surfaces and certain DAG cuts that avoid passing through a set of edges that we call \textit{unbreakable}. It allows us to generalize Thurston's theorem characterizing tilable regions by rewriting it in terms of descending paths or absorbing cycles. Thurston's algorithm appears as a distance calculation algorithm following Dijkstra's paradigm. The introduction of a set $X$ of interior edges introduces negative weights that force a Bellman-Ford strategy to be preferred. These results extend Thurston's legacy by using computer science structures and algorithms.
\end{abstract}

\section{Introduction}
Tilings have been a subject of interest for mathematicians for centuries, and more recently for famous mathematicians such as John Conway or William Thurston. Some of the most common tilings are tilings by \textit{calissons} i.e lozenges or rhombus. The name \textit{calisson} comes from the name of a French sweet made in Aix-en-Provence, a small town in the south of France.
Calisson tilings have the nice property to be interpreted in 3D as the perspective image of a stepped surface. 

\begin{figure}[ht]
  \begin{center}
		\includegraphics[width=0.7\textwidth]{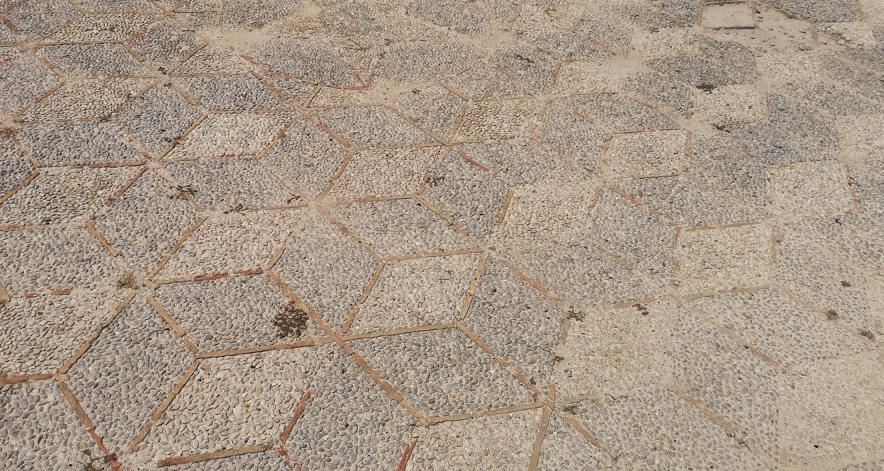}
	\end{center}
	\caption{\label{stabilo} \textbf{A calisson tiling} on the forecourt of the former building of \textit{Florio delle Tonnare di Favignana e Formica} on the island of Favignana where the 2022 edition of the excellent conference \textit{FUN with algorithms} took place. Calissons pavements give a 3d impression.}
\end{figure}

In this framework, Olivier Longuet, a french teacher of mathematics, created in 2022 an interesting logic puzzle called the \textit{Calissons Puzzle} (in french, the original name is \textit{le jeu des calissons}). This puzzle has the merit of developing children's sense of the third dimension and of being recreational.
A full description -in french- with many instances and an app to play online are available on a \href{https://mathix.org/calisson/blog/}{blog} led by Olivier Longuet. 
The rules are very simple. The problem is presented in a triangular grid bounded by a regular hexagon. 
A calisson is a pair of adjacent triangles. There are three types of calissons, each associated with a yellow, red or blue color, depending on their direction. 

\begin{figure}[ht]
  \begin{center}
		\includegraphics[width=0.95\textwidth]{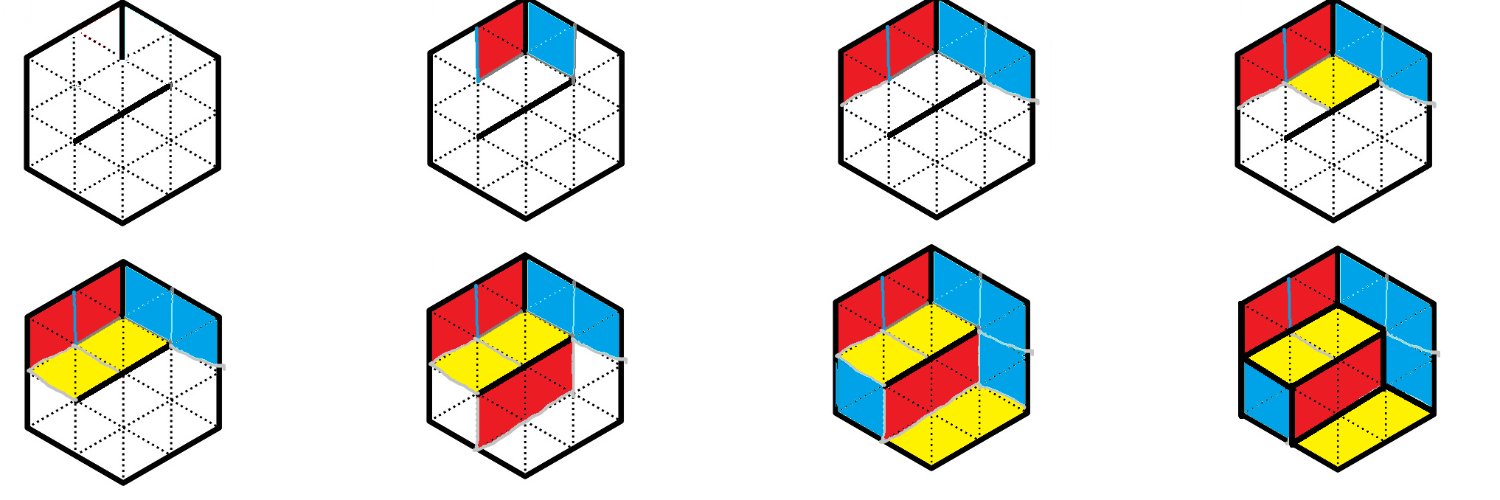}
	\end{center}
	\caption{\label{regles} \textbf{The rules of the puzzle} (image from Olivier Longuet's blog): we give ourselves a set of edges, for example in the top left-hand corner. The aim is to tile the hexagon with calissons in such a way that the edges given as input are adjacent to two calissons of different colors.}
\end{figure}

An instance of a calissons puzzle is made up of edges of the triangular grid. The problem is to tile the grid with calissons in such a way that the edges given as input are not overlapped by  a calisson and are adjacent to two calissons of different colors (Fig.~\ref{regles}). 

For a first try, two instances of the puzzle are drawn in figure~\ref{instance2}.

\begin{figure}[ht]
  \begin{center}
		\includegraphics[width=0.43\textwidth]{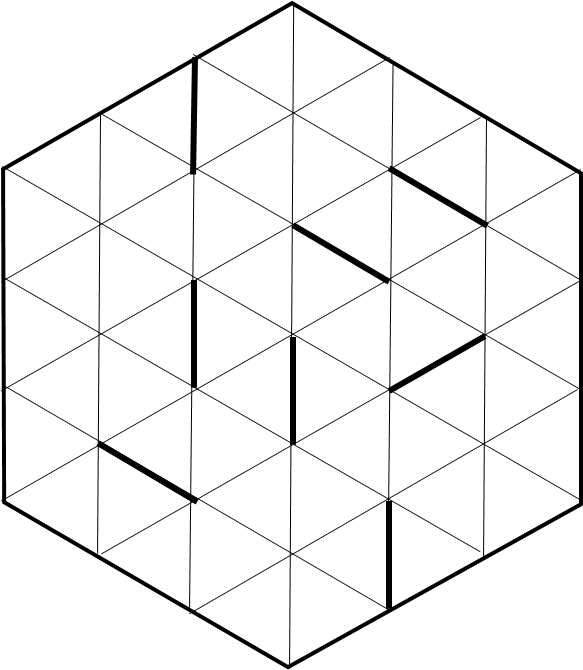}
  \hspace{1cm}
  \includegraphics[width=0.43\textwidth]{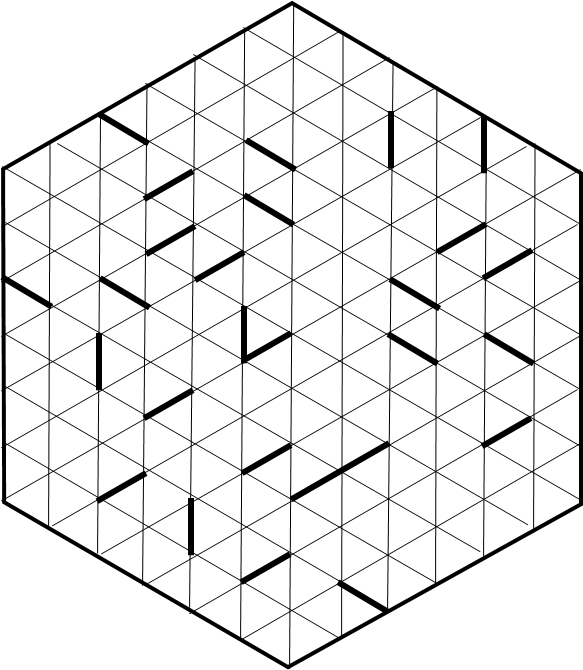}
	\end{center}
	\caption{\label{instance2} \textbf{Two instances of the puzzle.} The puzzle instance with $n=6$ is solved figure \ref{surfacequiavance}.}
\end{figure}

Our first goal is to determine the complexity of the puzzle.
We solve this question and a bit more in the triangular grid. 

\subsection{Notations}

The triangular grid can be defined as the projection of the cubic grid. 

\textbf{The grids $\square$ and $\square _n$.}  The primary cube  $C\subset \R ^3$ is $[0,1]^3$. The cubes of the cubic grid are simply denoted $(x,y,z)+C$ with $(x,y,z)\in \Z ^3$. These are the translates of $C$ by $(x,y,z)$. 
The sets of cubes, faces, edges and vertices of the cubic grid are respectively denoted $\square ^3$, $\square ^2$, $\square ^1$ 
and $\square ^0$ according to their dimension. Their union is a cubic complex denoted $\square = \square ^3 \cup \square ^2 \cup \square ^1 \cup \square ^0$.
For an integer $n$, we focus on the cellular complex $\square _n =\square ^3 _n \cup \square _n ^2 \cup \square _n ^1 \cup \square _n ^0$ containing the cubes, 
faces, edges and vertices of cubes $(x,y,z)+C$
where $(x,y,z)\in \{ 0 \cdots n-1\}^3$ with particular interest in the set of its cubes $\square _n ^3$. 

\textbf{The grids $\triangle$ and $\triangle _n$.} The infinite triangular grid $\triangle$ and its restriction $\triangle _n$ to the regular hexagon $\varphi([0,n]^3)$ are obtained by projecting the cell complexes $\square$ and $\square _n$ along $\varphi$ where $\varphi$ is the projection of the 3D space $\R^3$ onto a plane $H$ of equation $x+y+z=h$ in the direction $(1,1,1)$.

Rather than using two coordinates in the planar grid $\triangle$, the classic choice for working in the triangular grid is to use so-called \textit{homogeneous} coordinates. A point in the $\varphi (x,y,z)$ plane is identified by its three coordinates $(x,y,z)$, but to avoid any ambiguity, we  keep the letter $\varphi$ to differentiate between points in space noted $(x,y,z)$ and points $\varphi (x,y,z)$ in the  plane.  We obviously have $\varphi(x,y,z)=\varphi (x+k,y+k,z+k)$. Adding $k$ changes the depth of the  point in the $(1,1,1)$ direction without changing its projection. This notion of depth was put forward by the mathematician William Thurston under the name of \textit{height}, which we  use from now on, knowing that it is the height in the $(1,1,1)$ direction.

\begin{figure}[ht]
  \begin{center}
		\includegraphics[width=0.55\textwidth]{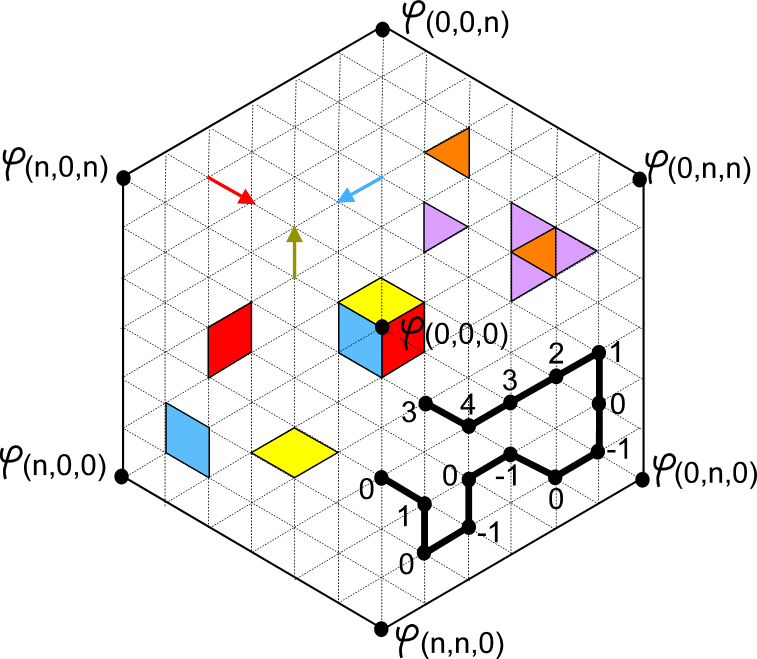}
	\end{center}
	\caption{\label{grid} \textbf{The triangular grid $\triangle _n$} for $n=6$. In the center, the vertex $\varphi (0,0,0)$ and the projection $\varphi(C)$ of the primary cube $C=[0,1]^3$. Above, edges of $\triangle ^1$ in the directions $\varphi(1,0,0)$, $\varphi(0,1,0)$ or $\varphi(0,0,1)$. On the right, a left and a right triangle and the three different ways to associate them in a calisson. On the left, a yellow, a red and a blue calisson. Below, a path $\delta$ and the heights $h(\delta _i)$ associated with points $\delta _i$ from an endpoint of height arbitrarily set at $0$.}
\end{figure}

The sets $\triangle ^0$ and $\triangle _n ^0$ of the vertices of the triangular grids $\triangle$ and $\triangle _n$ are respectively the projections of the vertices
of $\square$ and $\square _n$. The sets of edges $\triangle ^1$ and $\triangle _n ^1$ of the triangular grids $\triangle$ and $\triangle _n$ are the projections of the sets of edges $\square ^1$ and $\square _n ^1$.  From any vertex in $\triangle ^0$, we have six edges.  Their directions are $\varphi(1,0,0)$, $\varphi(0,1,0)$, $\varphi (0,0,1)$ and their opposite. 
The faces of the $\triangle$ and $\triangle _n$ grids, whose sets are $\triangle ^2$ and $\triangle ^2 _n$, are not projections of the faces of the $\square $ or $\square _n$ complexes, but triangles. We have two types of triangles. All have a vertical edge, but some point to the left and others to the right. We call them \textit{left} or \textit{right}. 

A calisson (or \textit{rhombus} or \textit{lozenge}) is the $\varphi$ projection of a face of the $\square$ grid. These are lozenges obtained by joining a left triangle to an adjacent right triangle of $\triangle$. As the faces of $\square$ have three directions, we have three types of calissons: 
blue, red and yellow calissons are respectively the projections of faces of normal direction $(1,0,0)$, $(0,1,0)$ and $(0,0,1)$. 
The set of calissons of the grids $\triangle$ and $\triangle_n$  are denoted $\Diamond$ and $\Diamond _n$. We have $\Diamond =\varphi (\square ^2)$ and $\Diamond _n=\varphi (\square _n ^2)$.

\subsection{Statements and Results}
With previous notations, original Olivier Longuet's calissons puzzle can be stated as follows.\\

\texttt{Calissons$(X,\triangle _n)$}\vspace{-0.26cm}
\begin{itemize}
    \item \texttt{Input:} An integer $n$ and a subset $X\subset \triangle ^1 _n$ of edges of the triangular grid.
    \item \texttt{Ouput:} a tiling of $\triangle_n$ by  $3n$ calissons so that (i) no edge of $X$ is ovelapped by the interior of a calisson and (ii) the two calissons adjacent to any edge of $X$ have different colors.  \\
\end{itemize}

Condition (i), called \textit{non-overlap condition}, is a natural condition in tiling definition. Condition (ii), that we call the \textit{saliency condition}, takes on its full meaning in dimension $3$, where it means that the edges of $X$ are salient edges of the staircase surface associated with the solution.

The initial problem we are interested in is to determine the complexity of the calissons puzzle. Passing through the notion of \textit{stepped surfaces} defined as a cut of a DAG, we show the following theorem.

\begin{theorem}\label{maincalissons}
An instance of the calissons puzzle \texttt{Calissons$(X,\triangle _n)$} can be solved with an algorithm of complexity $O(n^3)$.
\end{theorem}

The algorithm that we use is called the \textit{advancing surface}. It can be implemented directly on a printed puzzle with a pencil and a rubber. 
This first calissons puzzle is however a bit frustrating because there is no specific reason to be uniquely interested in tiling the triangular grid in the hexagon $\triangle _n$. This  class of hexagonal puzzles is however a warm-up before extending the puzzle to more general regions.
The extended version of the puzzle is denoted \texttt{Calissons$(X,R)$} where $R$ is the region to be tiled and $X$ is the set of imposed salient edges.\\

\texttt{Calissons$(X,R)$}\vspace{-0.26cm}
\begin{itemize}
    \item \texttt{Input:} A region $R\subset \triangle^2$ and a subset $X\subset \triangle ^1$ of edges of the triangular grid.
    \item \texttt{Ouput:} A calisson tiling of the region $R$  so that (i) no edge of $X$ is overlapped by the interior of a calisson and (ii) the two calissons adjacent to any edge of $X$ have different colors.  \\
\end{itemize}

We show how to solve this puzzle without using complex algorithms. The tools which allow us to solve it are even two of the most simple algorithms of graphs. They are the computation of a connective component and Bellman-Ford algorithm for computing the distances of the vertices of a graph from a source \cite{Bellman}.
It stems from the extremely simple structure of the calisson tilability problems that William Thurston highlighted in the early 1990s. 
We rewrite our general tilability problem \texttt{Calissons$(X,R)$} in three different ways in Theorem \ref{monte}. 
The exact statement requires notations introduced in the later, but  without going into the details, 
the existence of a solution of the extended calissons puzzle \texttt{Calissons$(X,R)$}  is equivalent to the existence of a cut in a graph itself equivalent to  the non-existence of a descending path, and at last to the non existence of an absorbing cycle in a weighted projected graph. 
The DAG cut formulation can be resolved by computing a connective component while the absorbing cycle can be detected with Bellman-Ford algorithm.
By solving the general tilability problem \texttt{Calissons$(X,R)$}, we revisit Thurston legacy  under the light of computer science with very classical structures of DAGs,  cuts, absorbing cycles and classical algorithms. 

We decompose the problem into two classes of instances depending on whether  the region $R$ is finite or not.
In the case where the region $R$ is simply connected and finite, we denote its boundary $\partial R$ and we generalize the previous advancing surface algorithm solving \texttt{Calissons$(X,\triangle _n)$} to \texttt{Calissons$(X,R)$}. It leads to the next result.

\begin{theorem}\label{wall}
Any instance of the extended calissons puzzle \texttt{Calissons$(X,R)$} for a finite, simply connected region $R$ can be solved with an algorithm of complexity $O(|\partial R |^3)$.
\end{theorem}

In the case of an unbounded region with no holes, the question is not to provide an explicit tiling of $R$ but to determine whether the instance admits a solution. The infinity of the region $R$ introduces a lock which is the computation of distances in an infinite graph. When this lock is open, as it is for the infinite triangular  grid $\triangle$, we use the absorbing cycle formulation to show the following result:

\begin{theorem}\label{wallinfini}
Any instance of the extended calissons puzzle \texttt{calissons$(X,\triangle)$} on the entire triangular grid $\triangle$ can be solved with an algorithm of complexity $O(|X|^3)$.
\end{theorem}

Following this introduction, the paper is organized into five sections. The section~\ref{Legacy} presents William Thurston legacy about the question of calisson tilability. 
The section~\ref{fail} shows that standard methods fail for solving the calissons puzzles. Then, contrary to usual practice, we do not present the general theory of \texttt{Calissons$(X,R)$} and then apply it to the particular case of calissons puzzles \texttt{Calissons$(X,\triangle _n)$}. We first present in Section~\ref{res} how to solve an instance of \texttt{Calissons$(X,\triangle _n)$}. The section~\ref{wwwalll} ends the paper with the extended version \texttt{Calissons$(X,R)$} and its resolution through equivalent propositions.

\section{Thurston's Legacy}\label{Legacy}

One of the questions explored by John Conway and William Thurston is whether a region is tilable by a given set of tiles, a question that applies to the triangular grid  $\triangle$ with calissons. John Conway gave an algebraic expression to the tilability problem by reducing it to the word problem. This problem consists in determining whether the word at the edge of the region represents the neutral element in the group generated by elementary displacements equipped with relations defined by the boundary of the tiles \cite{conway}.

Thurston revisited Conway's work  in the papers \cite{Thurston,Thurston2} by introducing a notion of height.

\subsection{Height in the triangular grid $\triangle$.}

Height is naturally defined in the three-dimensional space $\R ^3$. We define it according to the  direction $(1,1,1)$. The height of a point $(x,y,z)\in \square ^0$ is $h(x,y,z)=x+y+z$. We can not define the height of a point on the  grid $\triangle$ in an absolute manner, but we can define it in relative terms for points on a path.
Consider a path $\delta$ made up of consecutive points $\delta _i\in \triangle ^0$ linked by edges $\delta_i , \delta _{i+1} \in \triangle ^1$. This path can be lifted  in $\square$ to a path $\gamma$, a consecutive sequence of points $\gamma _i\in \square ^0$ such that $\varphi(\gamma _i)=\delta_i$ and $\gamma _i,\gamma _ {i+1}\in \square ^1$. This lift is not unique, as it can be made at different heights, but it is unique up to any vector translation $(k,k,k)$. The height differences between the points $\gamma _i$ are therefore independent of the chosen lift. If we set the height of $\gamma _0$ to $h(\gamma _0)=0$, we have a sequence of heights $h(\delta _i)$ defined by $h(\delta _i)=h(\gamma _i)$. The heights of the vertices on the $\delta$ path can be computed directly in the triangular grid. A step in the directions $-\varphi(1,0,0)$, $-\varphi(0,1,0)$, or $-\varphi(0,0,1)$  increases the height by $1$, while a step in the directions  $+\varphi(1,0,0)$, $+\varphi(0,1,0)$, or $+\varphi(0,0,1)$ decreases the height by $1$ (Fig.\ref{grid}).

\subsection{Tilability Characterization}

William Thurston has left his mark on problems involving the tilability of a region by calissons. We recall the two main results. The first  theorem characterizes simply connected regions $R$ tilable by calissons. 

\begin{theorem}[W. Thurston \cite{Thurston2}]\label{Thurston}
    A simply connected region $R\subset \triangle$ is tilable by calissons if and only if for any pair $u,v$ of vertices on the edge of $R$, we have $h(u)-h(v)\leq d(u,v)$ where $h$ denotes the height computed from a vertex on the edge of $R$ and where $d(u, v)$ is the distance between $u$ and $v$ in the graph with vertices $\triangle ^0 \cap R$ and edges oriented in the directions $-\varphi(1,0,0)$, $-\varphi(0,1,0)$ and $-\varphi(0,0,1)$. 
\end{theorem}

The second result is an optimal algorithm for determining whether a simply connected region can be tiled by calissons and providing a solution tiling if there exists one.

\subsection{Thurston's Algorithm}
The algorithm is illustrated Fig.~\ref{ThurstonAlgorithm}. It is a  beautiful algorithm simply based on heights computations. We decompose it in two steps.

\begin{figure}[ht]
  \begin{center}
		\includegraphics[width=0.85\textwidth]{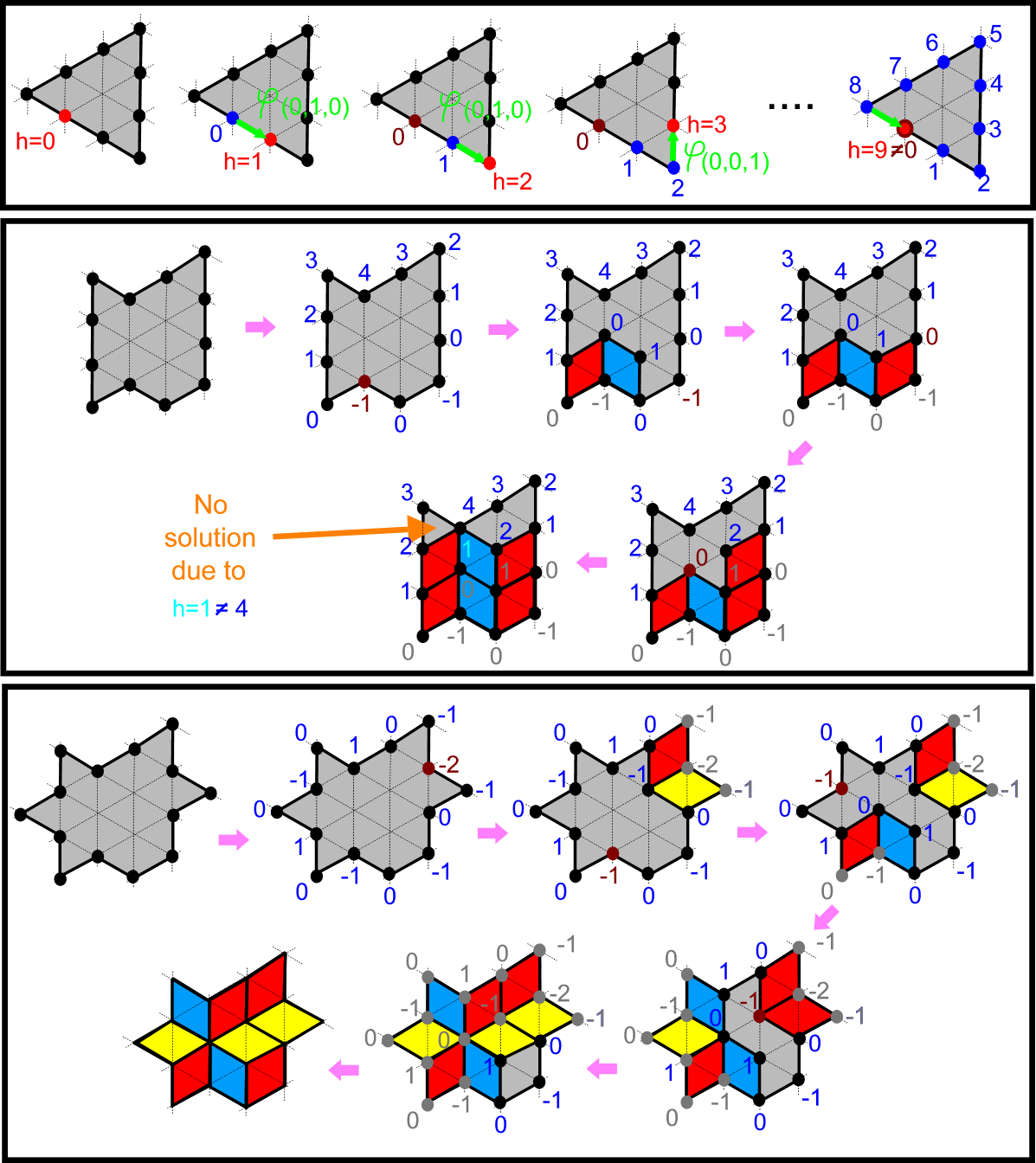}
	\end{center}
	\caption{\label{ThurstonAlgorithm} \textbf{Execution of Thurston's algorithm} on two instances where it shows that the region cannot be tiled, then a third instance for which it achieves the computation of the tiling of maximum height.}
\end{figure}

\begin{enumerate}
    \item Start from a vertex on the boundary $\partial R \subset \triangle ^0$  of the region $R$ to be tiled, and set its height to $0$. Then follow the edges of the boundary and increase the height by $1$ for a step $\varphi (1,0,0)$, $\varphi (0,1,0)$, $\varphi (0,0,1)$ or decrease it by $1$ for a step $-\varphi (1,0,0)$, $-\varphi (0,1,0)$, $-\varphi (0,0,1)$. 
    If, on returning to the starting point after the tour of $R$, the height is different from $0$, then the region $R$ is not tilable. If the height is $0$ after one turn, proceed to the next step.
    \item The second step consists in progressively tiling the region $R$ from its boundary.  The remaining region to be tiled is denoted $R'$ and its boundary $\partial R'$.     
    The algorithm repeats the following routine. Select a vertex $s$ of the path $\partial R'$ of minimum height. Tile it so that the vertices adjacent to $s$ in the tiling have a larger height. In other words, the edges of the new calisson(s) from $s$ must be directed by $\varphi (1,0,0)$, $\varphi (0,1,0)$ or $\varphi (0,0,1)$. Then compute the heights of the new vertices of $\partial R'$.
    Repeat the second step until one of the following two situations is reached:
\begin{itemize}
    \item An inconsistency arises because we want to overlap a vertex on the edge of $R$ with a new vertex of smaller height. In this case, according to Theorem \ref{Thurston}, there is no solution because we have $h(u)-h(v)\leq d(u,v)$ between two vertices $u$ and $v$ on the edge of $R$.
    \item In the second case, the  region $R$ is decimated until an empty $R'$ region is obtained. The region $R$ is tiled by calissons.
\end{itemize}
\end{enumerate}

We have a symmetrical version of the algorithm in which  vertices $s$ of maximum height are tiled with calissons whose edges are directed from $s$ by $-\varphi (1,0,0)$, $-\varphi (0,1,0)$, $-\varphi (0,0,1)$. 
These two versions of the algorithm respectively provide a maximum-height tiling and a minimum-height tiling. 

The complexity of Thurston's algorithm is linear in the size of the region ($O(|R|)$), i.e. linear in the size of the solution tiling. It is optimal.

For more details on domino and calisson tilings problems, apart from Thurston's work \cite{Thurston,Thurston2}, there is of course a large literature on the subject. See, for example, \cite{over} or Vadim Gorin's recent book \cite{gorin_2021}.\\

\textbf{What else ?} Thurston's results have a definitive character, as they elegantly and optimally solve a natural geometric problem. Nevertheless, we take on the challenge of revisiting them in the light of the calissons puzzle. 
The puzzle is more general than a simple tilability problem, it introduces other constraints and can be posed in an infinite region.
Thurston's algorithm cannot solve it.
This perspective, at the frontier of computer science and mathematics, with discrete structures and classical algorithms, provides an enlightening vision of the subject. It allows us to understand in depth the nature of Thurston's inheritance... and to extend it a little further.

\section{Matching and 3-SAT}\label{fail}

A reasonable idea for solving calissons puzzles is to use classical techniques from tîling problems. We already noticed that Thurston's algorithm cannot take account of the interior edges of $X$, nor of saliency constraints. It is therefore unable to solve the calissons puzzles. 

However, there are other approaches, either used for tilability by dominos
or for general combinatorial problems. Two methods are worth examining. The first reduces the problem to $3$-SAT, while the other involves the computation of a matching in a bipartite graph.

\subsection{3-SAT}

The calissons puzzle is easily expressed as a 3-SAT formula.
Consider a variable $a_c$ for each calisson $c$ in $\Diamond _n$. It is equal to $1$ if the calisson $c$ is included in the solution's tiling and $0$ otherwise.

We have four classes of clauses.

\begin{enumerate}
    \item The first clauses express the conditions that all triangles of $\triangle ^2 _n$ must be covered by at least one calisson. This constraint is expressed in the form of 3-clauses, since there are no more than three calissons covering a triangle. For each triangle $t\in \triangle ^2 _n$, we impose $a_c \vee a_{c'} \vee a_{c''}$ where $c$, $c'$ and $c''$ are the calissons covering the triangle $t$  (for boundary triangles, these are 2-clauses and even 1-clauses).
    \item The second class of clauses is still necessary to guarantee that we have a tiling: the tiles must not overlap. For each pair $c,c'$ of calissons with a triangle in common, we impose $\overline a_c \vee \overline a_{c'}$ to ensure that there do not overlap.
    \item The third class of clauses expresses the non-overlap constraint (i) of the puzzle. Some variables are set to $0$.
    \item The last class of clauses expresses the saliency constraints (ii). 
    Around an edge for instance covered by the interior of a yellow calisson, the red calisson $c$ on one side imposes a blue calisson $c'$ on the other, and vice-versa. We thus have clauses  $\overline a_c \vee a_{c'}$.
\end{enumerate}

The number of variables and clauses is $O(n^2)$. It provides a simple way of expressing the problem and solving it with a solver. As the 3-SAT problem is NP-complete, this reduction does not allow us to solve the calissons puzzle in polynomial time.

\subsection{Matching}

 A classic, non-exponential approach to compute tilings by dominoes or \textit{dimers} (and calissons are dominoes made up of two adjacent triangles) is to compute a matching in the adjacency graph of triangles (see for example \cite{remila}). This approach is illustrated in Fig.~\ref{mat}.

\begin{figure}[ht]
  \begin{center}
	\includegraphics[width=0.95\textwidth]{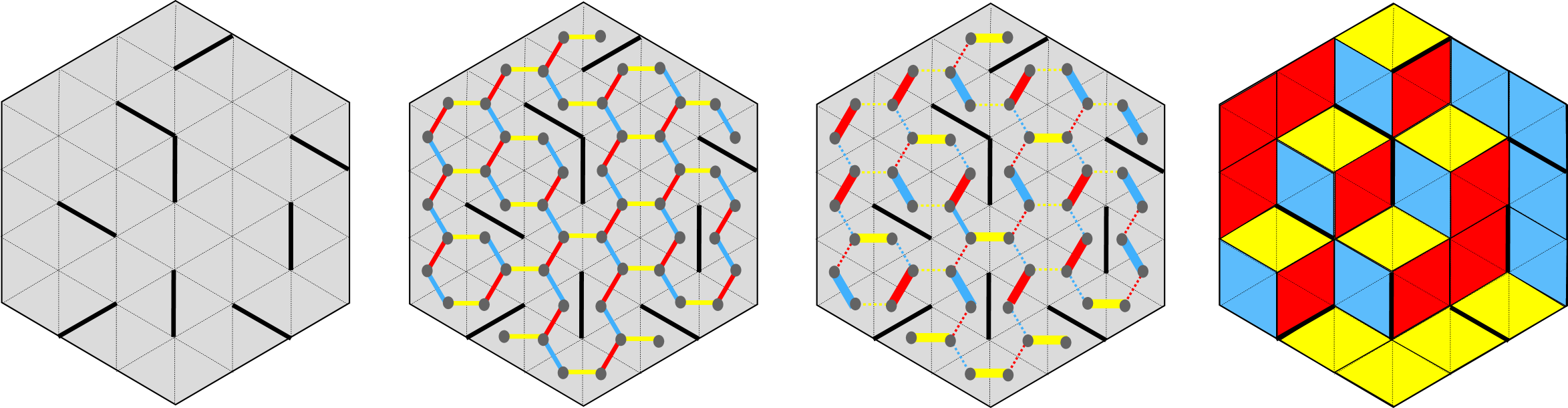}
	\end{center}
	\caption{\label{mat} \textbf{Try to solve a puzzle through a matching computation}. Left, an instance of the calissons puzzle takes the form of a set of edges $X$ of the triangular grid $\triangle _n$. In the center, the adjacency graph of the triangles and a perfect matching of $\triangle_n ^2$. On the right, the computed tiling  satisfies the non-overlap condition (i) but violates the saliency condition (ii).  }
\end{figure}

From the set of edges $X$ of the calissons puzzle instance, we create the graph $\Gamma$ whose vertices are the triangles of the grid and whose edges are the pairs of adjacent triangles that are not separated by an edge of $X$. A perfect matching of the $\Gamma$ graph is computed. If there is no solution, the calissons instance admits no solution. If there is a perfect matching $M$ of $\Gamma$, then $M$ provides a tiling of $\triangle _n$ that satisfies the edge non-overlap rule (i) but may violate the saliency conditions (ii), as is the case on the right-hand side of the example shown in Fig.~\ref{mat}.

To take account of saliency constraints (ii), we might want to adapt the matching algorithms so as to guarantee that if one edge is chosen, so is another. However, this seems unrealistic, as it can easily be shown that such associations harden the matching problem. The problem of computing intersection-free matching in a geometric graph, for example, is NP-hard, which is all the more detrimental as we can easily reduce \texttt{Calissons$(X,\triangle _n)$} to this problem. 
In other words, the matching approach does not solve the calissons puzzle. 

\section{The Advancing Surface Algorithm for solving \texttt{Calissons$(X,\triangle _n)$}}\label{step}

For solving the calissons puzzles \texttt{Calissons$(X,\triangle _n)$}, we start by introducing the 3D notion of \textit{stepped surface} of $\triangle _n$ (term used in \cite{Fernique}). We define them as the cuts of a DAG of vertices in $\square_n$. Then we express the constraints induced by the non-overlap rule (i) and the saliency conditions (ii) on the DAG in order to analyse the problem according to this perspective. 

\subsection{Stepped Surface of $\triangle _n$ as DAG cuts}
We first introduce the stepped surfaces above $\triangle _n$.
We complete the set of cubes $\square ^3 _n$ (its cubes are $(x,y,z)+C$ with $0\leq x \leq n-1$, $0\leq y \leq n-1$, $0\leq z \leq n-1$) with two other sets  denoted $\back _n$ and $\front _n$. The set $\back _n$ contains the cubes $(x,y,z)+C$ with two integral coordinates between $0$ and $n-1$ and the last coordinate equal to $-1$.
The set $\front _n$ contains the cubes $(x,y,z)+C$ with two integral coordinates between $0$ and $n-1$ and the last coordinate equal to $n$.

Then we introduce a first DAG structure $\H=(\square ^3,\wedge)$ on the whole set of cubes $\square ^3$ with an edge from any cube $(x,y,z)+C$ to the cubes  $(x+1,y,z)+C$, $(x,y+1,z)+C$ and $(x,y,z+1)+C$. We use the notation $\wedge$ for the set of edges since they are ascendant according to the height $x+y+z$. Notice that each edge of $\H$ can be represented geometrically by the common face of the two cubes. 

We denote $\H _n$ the induced graph of $\H$ on the set of vertices $\back _n \cup \square _n \cup \front _n$. In other words, we have the DAG $\H _n=(\back  _n\cup \square ^3_n \cup \front _n, \wedge)$. The transitive closure of $\H _n$ is a partial ordered set (poset). This partial order relation is denoted $(x,y,z)+C \leq (x',y',z')+C $ so that we have $(x,y,z)+C \leq (x',y',z')+C $ if and only if $x\leq x'$ and $y\leq y'$ and $z\leq z'$. Incomparable cubes are denoted by $(x,y,z)+C \sim (x',y',z')+C $.

\begin{figure}[ht]
  \begin{center}
		\includegraphics[width=0.95\textwidth]{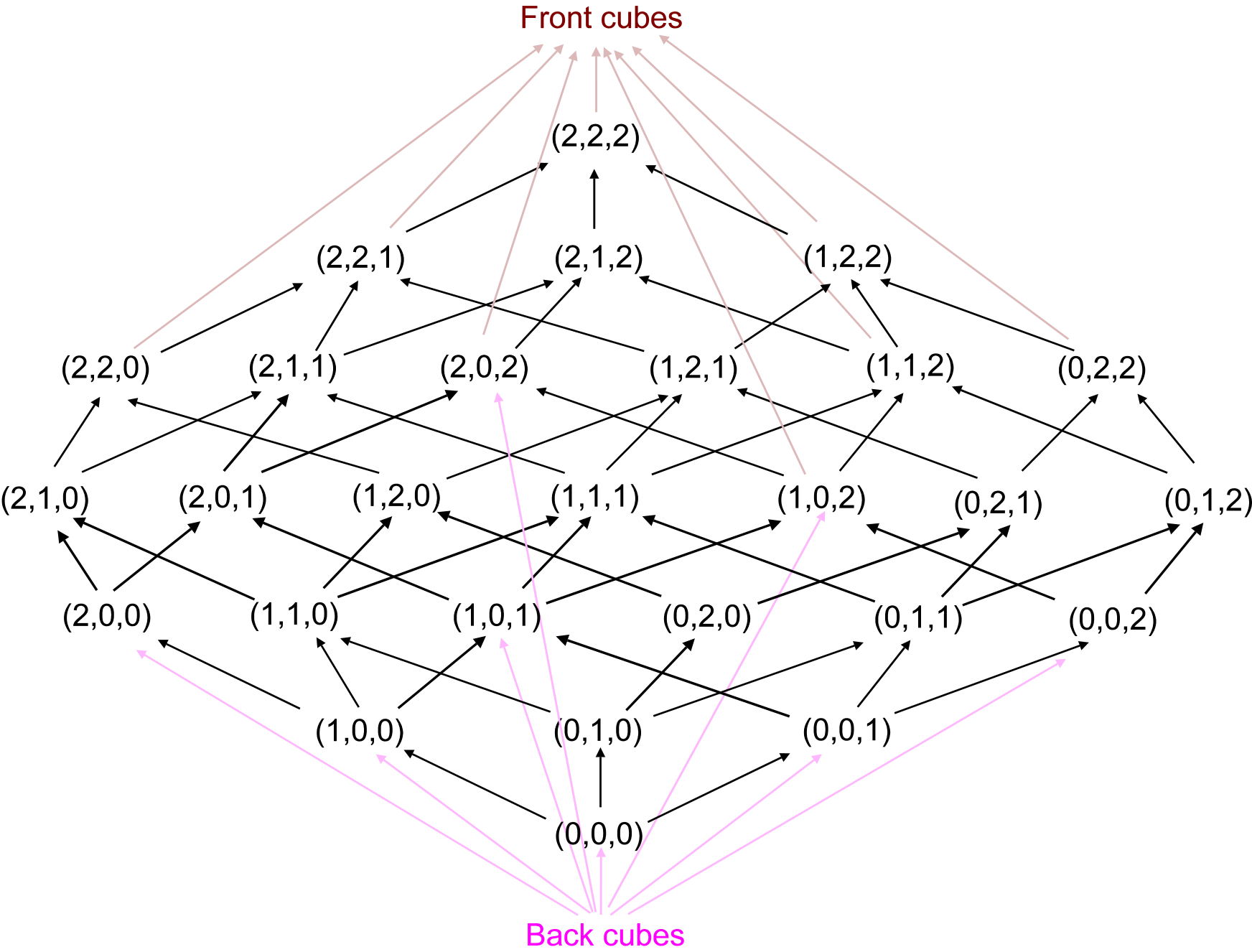}
	\end{center}
	\caption{\label{Hasse} \textbf{The DAG $\H _n$ of set of vertices $ \mathrm{Back} _n \cup \square _n \cup \mathrm{Front} _n$} for $n=3$. The cubes of $\mathrm{Back} _n$ and $\mathrm{Front} _n$ are not individually represented, neither all the edges issued or arriving to their vertices. Each edge of the DAG $\H_n$ corresponds to the common face of a pair of cubes  }
\end{figure}

\textbf{DAG Cut.}
 There is a general notion of a \textit{cut in a graph} that we call \textit{graph cut}. It is a partition of the set of vertices into two parts, and we are particularly interested in the edges going from one part to the other. 
There is another notion of a cut in a poset or DAG which is  more restricted and that we call equivalently \textit{DAG cut} or \textit{poset cut}. 
In a poset, a set is said to be \textit{low} if it is the union of all elements less than or equal to its elements, and \textit{high} if it is the union of all elements greater than or equal to its \cite{Abian} elements.  Given a low  part  of a poset, its complementary is necessarily high, and vice versa. 
A \textit{poset cut} is then a non-trivial partition (no empty parts), of the set of vertices into a lower part $L$ and an upper part $H$. A \textit{DAG cut} of a DAG $\Gamma$ is the poset cut of the transitive closure of $\Gamma$. Rather than focusing on the subsets $L$ and $H$, it is natural to look at the edges of the DAG from $L$ to $H$. 

\begin{definition}
A \textit{stepped surface} of $\triangle _n$ is the set of the edges $E\subset \wedge$ of the DAG $\H =(\square ^3 , \wedge)$  going from the lower part to the upper part of a DAG cut of $\H _n$ separating $\back _n$ from $\front _n$  (Fig.~\ref{cutHasse}).
\end{definition}

\begin{figure}[ht]
  \begin{center}
		\includegraphics[width=0.95\textwidth]{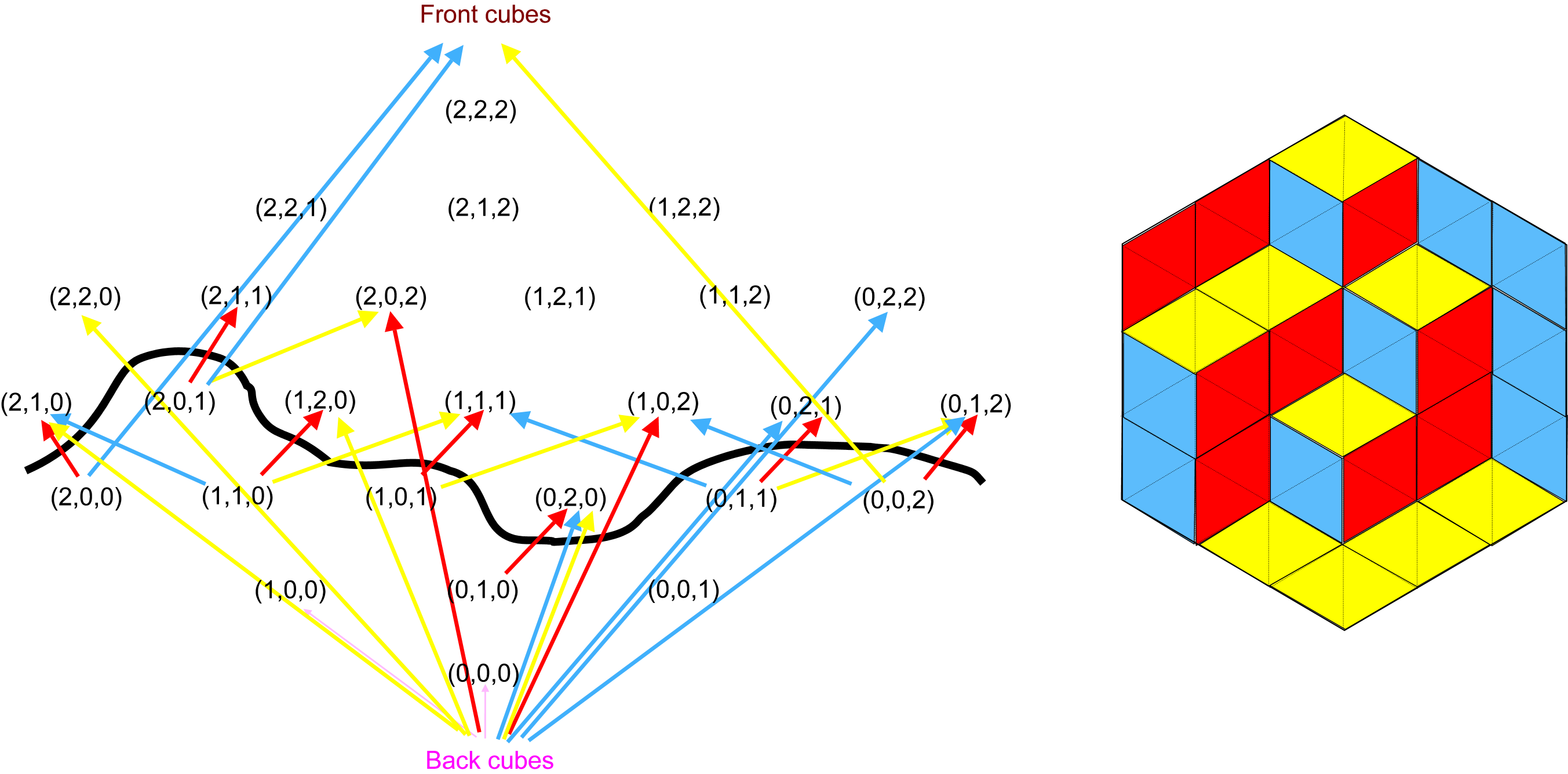}
	\end{center}
	\caption{\label{cutHasse} \textbf{A cut of the DAG $\H _n$  and the calisson tiling of the corresponding stepped surface}. }
\end{figure}

The projection $\varphi$ is a one-to-one map between the stepped surfaces and the calisson tilings of  $\triangle _n$. This theorem can be seen as folklore. We do not prove it but a close theorem -Theorem~\ref{monte}- relating tilings and cuts is proved in the later.  


\subsection{Constraints}\label{res}

Thanks to the one-to-one map $\varphi$ between calisson tilings and stepped surface,  the calissons puzzle consists in determining a stepped surface that satisfies the constraint (i) of not overlapping the edges of $X$ and the saliency constraint (ii). 

Given an edge $e$ in $X$, what is the condition on the DAG cuts of the constraints (i) and (ii) imposed by $e$? 
The translation of these constraints onto a stepped surface can be expressed through the following lemmas:

\begin{lemma}\label{constraintz}
We consider a vertical edge $e=\varphi(x,y,z),\varphi(x,y,z+1)\in \triangle ^1 _n$. The cubes of $ \square ^3$ one of whose projected faces is adjacent or overlapping $e$ are denoted  $L_k=(x+k,y+k-1,z+k)+C$, $R_{k}=(x+k-1,y+k,z+k)+C$,
$F_k=(x+k,y+k,z+k)+C$ and $B_{k}=(x+k-1,y+k-1,z+k)+C$ (Fig.~\ref{LRBF}). 

The calisson tiling of the stepped surface $S$ satisfies the non overlapping constraint (i) of $e$ if and only if the  cut $S$  does not separate a pair of cubes $F_k$ and $B_{k+1}$.

The calisson tiling of the stepped surface $S$ satisfies the saliency  constraint (ii) of $e$ if and only if the cut $S$  separates neither a pair of cubes $F_k$ and $B_{k+1}$ (this is constraint (i)), nor a pair of cubes $L_k$ and $R_{k}$.

The oriented edges $F_k\rightarrow B_{k+1}$, $B_{k+1} \rightarrow F_k$, $L_k \rightarrow R_{k}$ and $R_k \rightarrow L_{k}$ of $\H$ are said \textit{unbreakable}.
\end{lemma}

\begin{proof}
The key point is that the union of the four cube sequences $L_k$, $R_k$, $B_k$, $R_k$ (Fig.~\ref{LRBF}) is almost totally ordered according to the partial order relation $\leq$ given by the transitive closure of the DAG $\H =(\square ^3  , \wedge )$. Only $L_k$ and $R_k$ are incomparable. 
Then we have 
$$\cdots \leq F_{k-1} \leq B_{k} \leq L_{k} \sim R_{k} \leq F_{k} \leq B_{k+1} \leq L_{k+1} \sim R_{k+1} \leq  F_{k+1} \leq \cdots$$ 

\begin{figure}[ht]
  \begin{center}
		\includegraphics[width=0.65\textwidth]{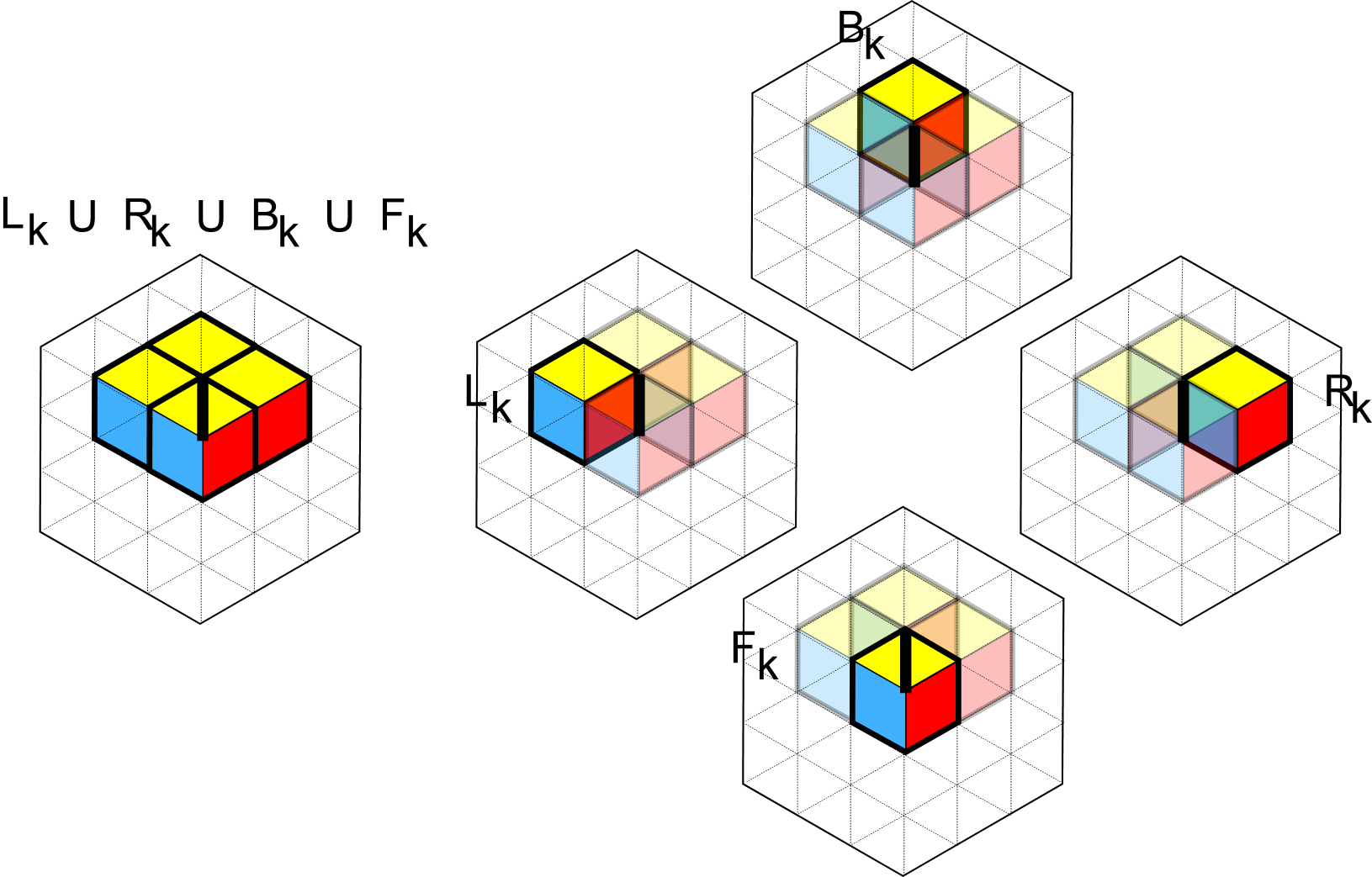}
	\end{center}
	\caption{\label{LRBF} \textbf{The $B_k$, $L_k$, $R_k$ and $F_k$ cube} for a given height $k$. They have a face $f$ whose projection $\varphi (f)$ is a calisson adjacent to or overlapping the edge $e$.}
\end{figure}

We have a chain of calissons 
and any stepped surface intersects it at a certain level. Within one index shift, we have four different DAG cut cases, illustrated in Fig.\ref{fourcases} and each giving a different configuration around the  edge $e$:
\begin{enumerate}
    \item The DAG cut separates $B_k$ and the two cubes $L_{k} \sim R_{k}$. In this case, the stepped surface contains the face common to $B_k$ and $L_k$ and the face common to $B_k$ and $R_k$. These are the two faces adjacent to $e$ and they are of different colors. Conditions (i) and (ii) are satisfied.    \item The DAG cut separates the two cubes $L_{k}$ and $R_{k}$. We have the sub-case where $L_{k}$ is under the DAG cut/behind the surface and $R_{k}$ is in front of the stepped surface. In this sub-case 2, the stepped surface contains the face common to $L_k$ and $F_k$ and the face common to $B_k$ and $R_k$. These are the two faces adjacent to $e$ and they are both red.
    Then there's the sub-case where $R_{k}$ is under the DAG cut/behind the surface and $L_{k}$ is in front of the stepped surface. In this sub-case 2', the stepped surface contains the face common to $B_k$ and $L_k$ and the face common to $R_k$ and $F_k$. These are the two faces adjacent to $e$ and they are both blue.
    In these two sub-cases, condition (i) is satisfied and condition (ii) is violated.
    \item The DAG cut separates the two cubes $L_{k} \sim R_{k}$ from $F_k$. In this case, the stepped surface contains the face common to $L_k$ and $F_k$ and the face common to $R_k$ and $F_k$. These are the two faces adjacent to $e$ and they are of different colors. In this case, both conditions (i) and (ii) are satisfied.
    \item The DAG cut separates $F_k$ and $B_{k+1}$. In this case, the stepped surface contains the face $f$ common to $F_k$ and $B_{k+1}$. The projected calisson $\varphi (f)$ of this face overlaps the edge $e$. In this case, condition (i) is violated.
\end{enumerate}

\begin{figure}[ht]
  \begin{center}
		\includegraphics[width=0.95\textwidth]{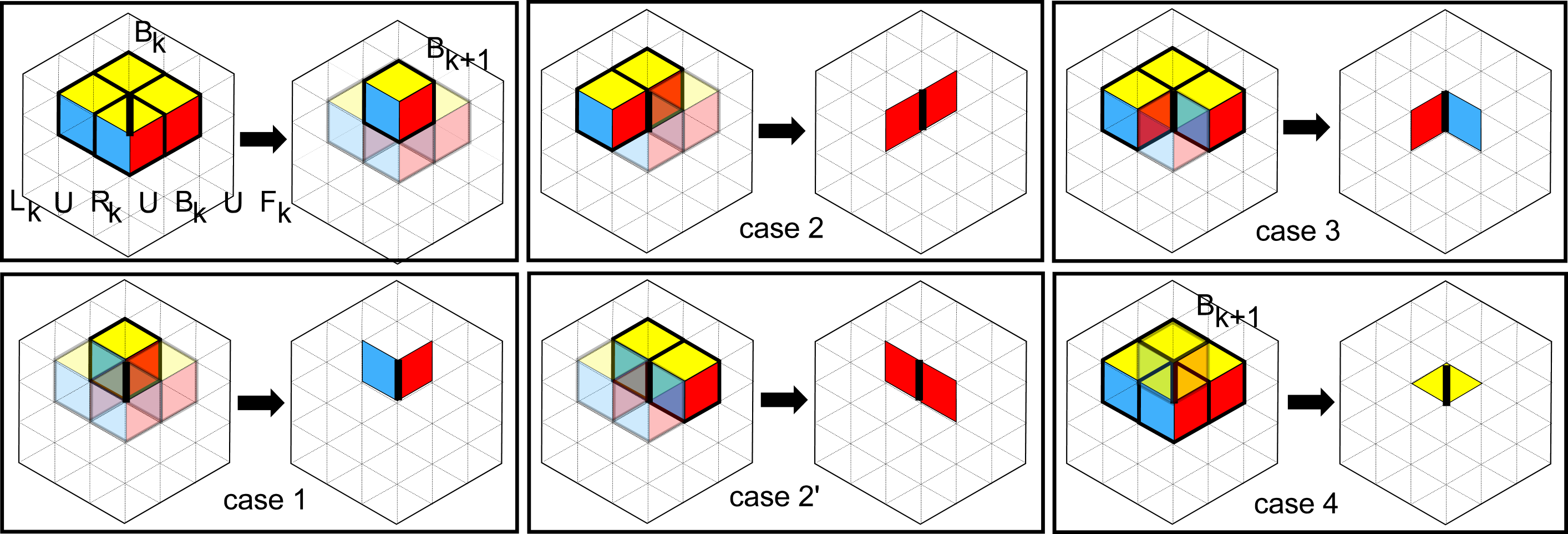}
	\end{center}
	\caption{\label{fourcases} \textbf{Tiling configurations around the edge $e$.} We have four different cases depending on where the DAG cut of $\H$ cuts the almost totally ordered sequence of cubes $\cdots \leq B_{k} \leq L_{k} \sim R_{k} \leq F_{k} \leq B_{k+1}  <\cdots$. }
\end{figure}

\end{proof}

We have similar lemmas for non-vertical edges. 

\begin{lemma}\label{constraintx}
We consider an edge $e=(x,y,z),(x+1,y,z)\in \triangle ^1 _n$. The cubes of $ \square ^3 $ one of whose projected faces is adjacent or overlapping $e$ are denoted $L_k=(x+k,y+k-1,z+k)+C$, $R_{k}=(x+k,y+k,z+k-1)+C$,
$F_k=(x+k,y+k,z+k)+C$ and $B_{k}=(x+k,y+k-1,z+k-1)+C$. 

The calisson tiling of the stepped surface $S$ satisfies constraint (i) of $e$ if and only if it does not separate a pair of cubes $F_k$ and $B_{k+1}$.

The calisson tiling of the stepped surface $S$ satisfies constraint (ii) of $e$ if and only if it separates neither a pair of cubes $F_k$ and $B_{k+1}$ (this is constraint (i)), nor a pair of cubes $L_k$ and $R_{k}$.

The oriented edges $F_k \rightarrow B_{k+1}$, $B_{k+1} \rightarrow F_k$, $L_k \rightarrow R_{k}$ and $R_k \rightarrow L_{k}$ of $\H$ are said to be \textit{unbreakable}.
\end{lemma}

\begin{lemma}\label{constrainty}
We consider an edge  $e=(x,y,z),(x,y+1,z)\in \triangle ^1 _n$ . The cubes of $ \square ^3 $ one of whose projected faces is adjacent or overlapping $e$ are denoted  $L_k=(x+k-1,y+k,z+k)+C$, $R_{k}=(x+k,y+k,z+k-1)+C$,
$F_k=(x+k,y+k,z+k)+C$ and $B_{k}=(x+k-1,y+k,z+k-1)+C$. 

The calisson tiling of the stepped surface $S$ satisfies constraint (i) of $e$ if and only if it does not separate a pair of cubes $F_k$ and $B_{k+1}$.

The calisson tiling of the stepped surface $S$ satisfies constraint (ii) of $e$ if and only if it separates neither a pair of cubes $F_k$ and $B_{k+1}$ (this is constraint (i)), nor a pair of cubes $L_k$ and $R_{k}$.

The oriented edges $F_k \rightarrow B_{k+1}$, $B_{k+1} \rightarrow F_k$, $L_k \rightarrow R_{k}$ and $R_k \rightarrow L_{k}$ of $\H$ are said to be \textit{unbreakable}.
\end{lemma}

We introduce a few notations to denote the sets of unbreakable edges. Given a set of edges $X\subset \triangle ^1$, we denote $\Pi _{(i)}(X)$ the set of unbreakable edges $F_k\rightarrow B_{k+1}$ and $B_{k+1} \rightarrow F_k$ with $k\in \Z$. They express the non-overlap constraint (i).
Let $\Pi _{(ii)}(X)$ denote the set of unbreakable edges $L_k \rightarrow R_{k}$ and $R_k \rightarrow L_{k}$ with $k\in \Z$, which express the saliency constraints (ii) of the edges of $X$. 
Finally, we denote $\Pi (X) = \Pi _{(i)}(X) \cup \Pi _{(ii)}(X)$ the set of all unbreakable edges imposed by $X$.

Given an instance $X$ of the calissons puzzle, the set $\Pi (X)$ contains directed edges with vertices in $\square ^3$. 
According to the lemmas \ref{constraintz}, \ref{constraintx}, \ref{constrainty}, the edges of $\Pi (X)$ do not have to be cut to obtain a stepped surface solution of the puzzle.

\subsection{Reduction}

By considering the calisson tilings as DAG cuts of $\H_n$, the lemmas \ref{constraintz}, \ref{constraintx}, \ref{constrainty} prove the following theorem.

\begin{theorem}\label{insecable}
    An instance \texttt{Calissons}$(X,n)$ admits a solution if and only if the DAG $\H_n$ has a DAG cut separating $\back _n$ from $\front _n$ and cutting no unbreakable edge of $\Pi (X)$.
\end{theorem}

The theorem  \ref{insecable} reduces the calissons puzzle to the computation of a DAG cut of $\H_n$. Three examples one with a solution and two without are illustrated  Fig.~\ref{insecable1}, Fig.~\ref{insecable2} and Fig.~\ref{insecable3}.

\begin{figure}[ht]
  \begin{center}
		\includegraphics[width=0.95\textwidth]{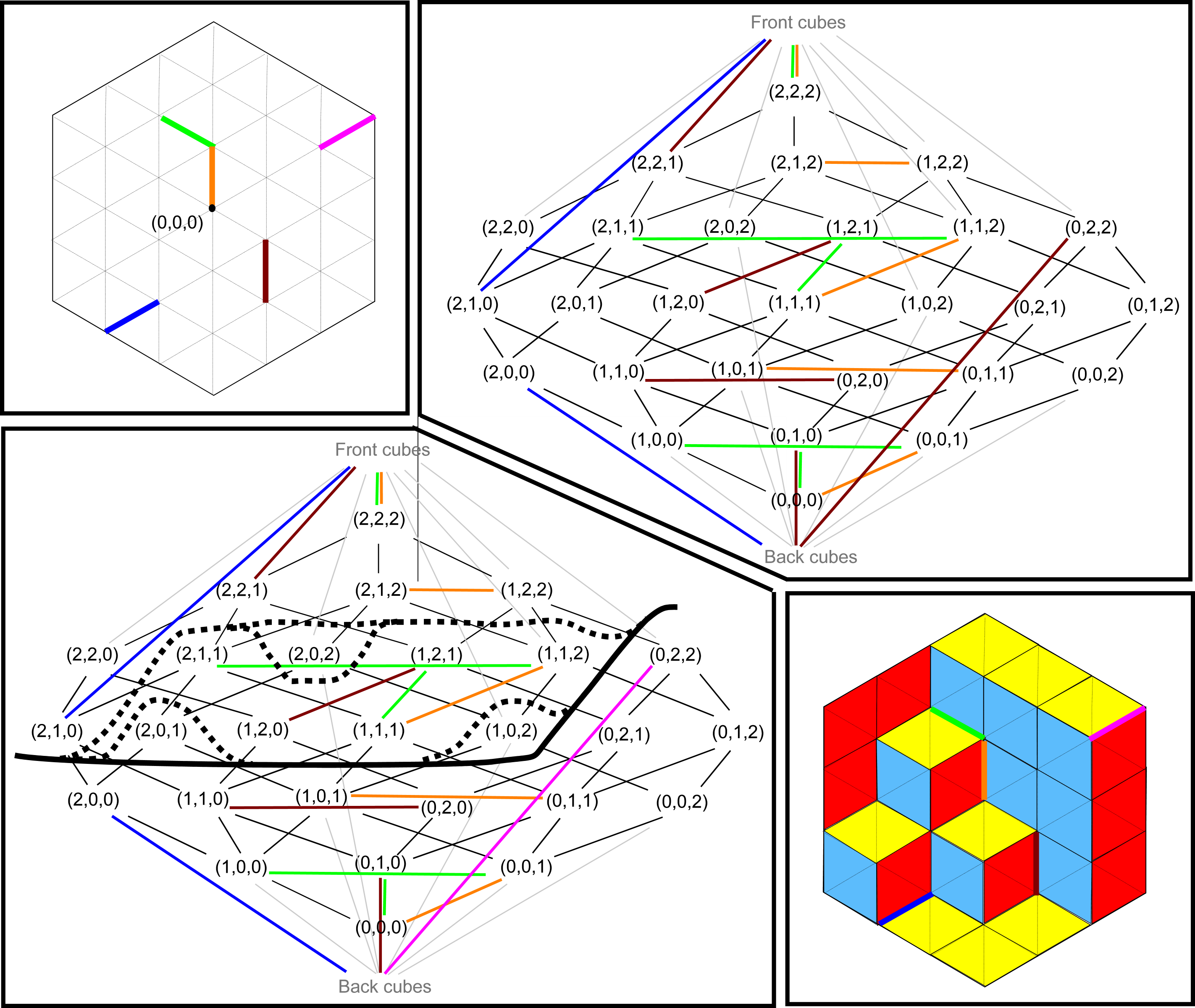}
	\end{center}
	\caption{\label{insecable1} \textbf{Solving calissons puzzles} can be reduced to the computation of a DAG cut of  $\H _n$ which does not cut unbreakable edges (these edges have the color of the  edge $e$ from which they originate). Top left, an instance of a puzzle. Top right, the unbreakable edges of $\H _n$. Bottom, a DAG cut that does not cut an unbreakable edge (and its dotted alternatives) and the corresponding tiling.}
\end{figure}

\begin{figure}[ht]
  \begin{center}
		\includegraphics[width=0.95\textwidth]{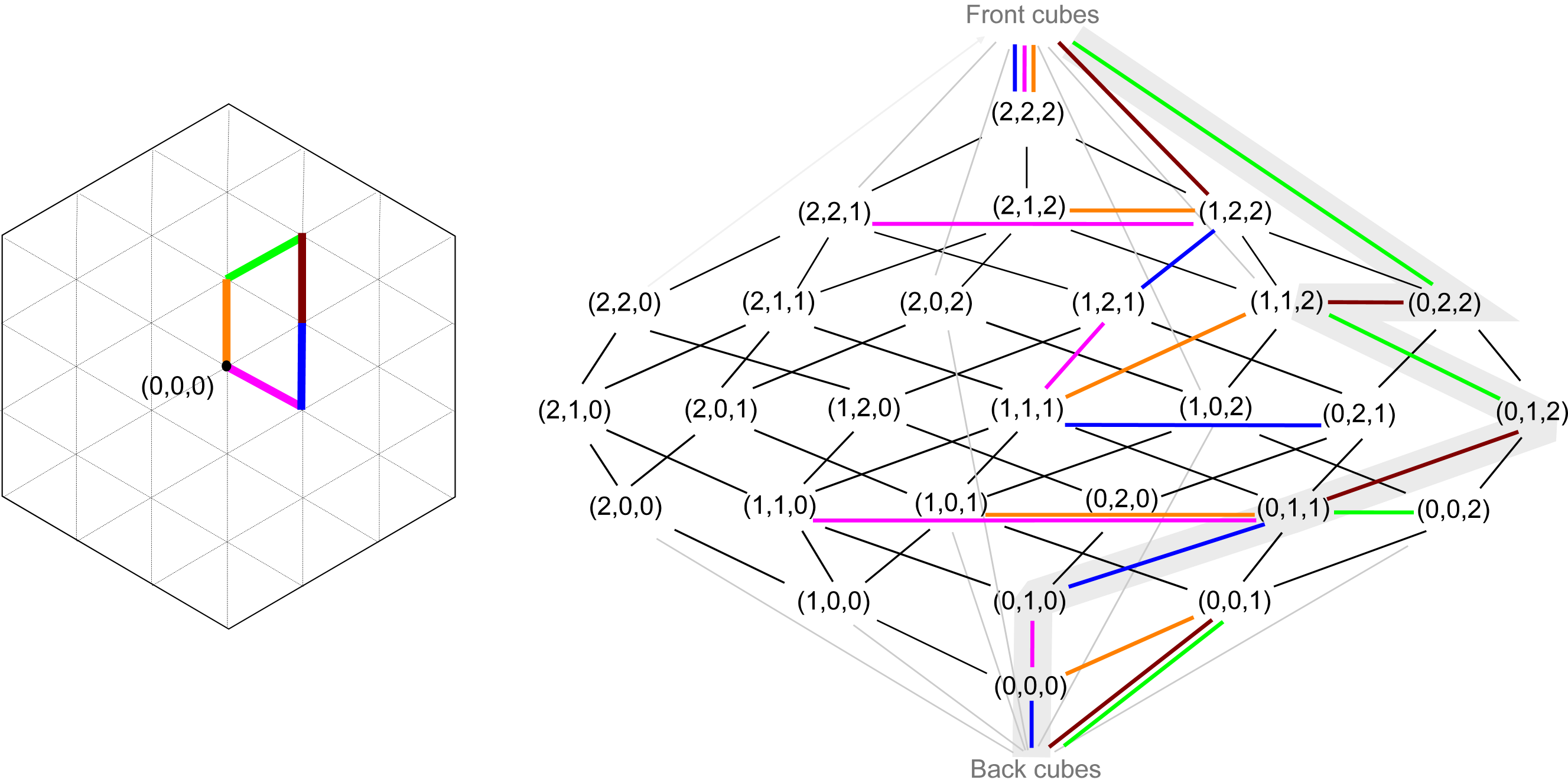}
	\end{center}
	\caption{\label{insecable2} \textbf{Example with no solution.} The absence of a solution is trivial on the tiling because a path surrounds an odd number of triangles and this results in the existence of an unbreakable edge path linking $\back _n$ to $\front _n$ in DAG $\H_n$, which prevents $\back _n$ from being separated from $\front _n$.}
\end{figure}

 \begin{figure}[ht]
  \begin{center}
		\includegraphics[width=0.95\textwidth]{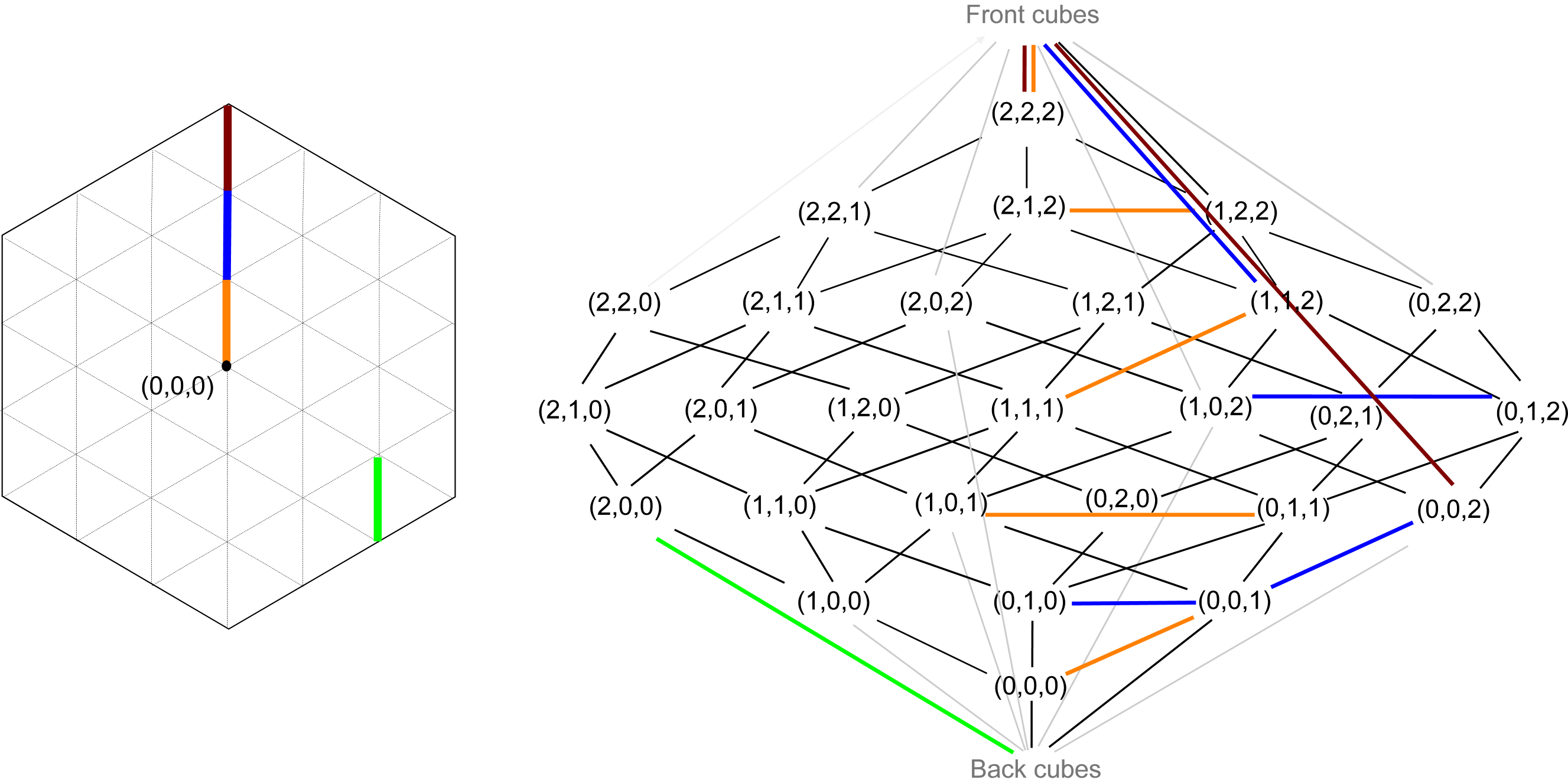}
	\end{center}
	\caption{\label{insecable3} \textbf{Example with no solution and no path connecting $\front _n$ to $\back _n$}. If there is an unbreakable path connecting $\front _n$ from $\back _n$, then the puzzle instance admits no solution, but the converse is false. To get the equivalence, the DAG $\H_n$ has to be completed by the set of unbreakable edges.}
\end{figure}

The computation of graph cuts is a classical algorithmic problem. DAG cuts are a bit different due to the  constraint to separate a low from a high set of vertices.  

\begin{figure}[ht]
  \begin{center}
		\includegraphics[width=0.95\textwidth]{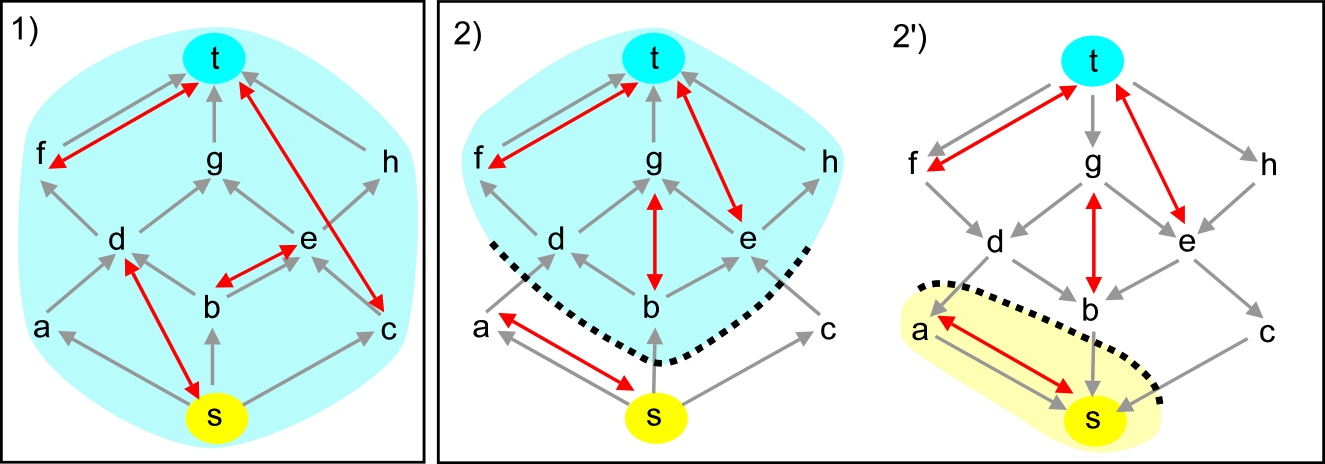}
	\end{center}
	\caption{\label{compo} \textbf{Search for a DAG cut of $(V,E)$ that does not cut an unbreakable edge}. Unbreakable non oriented edges are drawn in red and DAG edges in grey. To find a DAG cut that separates $s$ and $t$ and avoids cutting the unbreakable edges, we complete the DAG $(V,E)$ with the  set $\Pi$ of unbreakable non oriented edges. In 1) there are graph cuts separating $s$ and $t$ without cutting any unbreakable edges, but there are no DAG cuts, since the sets $P_S$ and $P_T$ of a graph cut are neither low nor high. It is impossible to separate a low part containing $s$ and a high part containing $t$ without cutting an unbreakable edge, as the connective component of $t$ in the completed graph $(V,E\cup \Pi)$ (component in light blue) contains $s$. In 2), the connective component of $t$ in the completed graph $(V,E\cup \Pi)$ does not contain $s$. It provides the upper part $P_T$, which is not separated from its complementary (lower) part by any unbreakable edges. By changing the direction of the DAG edges and keeping the unbreakable edges in 2'), the connective component of $s$ provides the lowest DAG cut.}
\end{figure}

We show now how to solve  the DAG cutting problem in a DAG $(V,E)$ by avoiding to cut a set of unbreakable edges denoted $\Pi$. 
It is assumed that the part $P_S$ containing $S$ is destined to be the low/source part and its complementary $P_T$ the high/terminal part. The search for a DAG cut separating a low part containing $S$ from a high part containing $T$ is solved by computing the connective component of $T$ in the graph $(V,E \cup \Pi)$ where the initial DAG $(V,E)$ is completed  with the set $\Pi$ of unbreakable edges (Fig.~\ref{compo}). If the connective component of $T$ contains a vertex of $S$, then there is no valid DAG cut. If the connective component of $T$ does not contain a vertex of $S$, then this set of vertices together with its complement provides the highest valid DAG cut.

We can also reverse the direction of the DAG edges and compute the connective component of $S$. If it contains $T$, there is no valid DAG cut. Otherwise, the connective component of $S$ together with its complement provides the lowest valid DAG cut (Fig.~\ref{compo}). 

Applying Theorem \ref{insecable} and the computation of a DAG cut with unbreakable edges in a DAG, we reduce the computation of a solution to the calissons puzzle to the computation of the connective component of $\front _n$ in the DAG 
$\H _n$ completed by the set $\Pi (X)$ of unbreakable edges.
If the connective component of $\front _n$ contains a cube of $\back _n$, the puzzle has no solution. Otherwise, the connective component of $\front _n$ provides a valid DAG cut, i.e. a calisson tiling satisfying the puzzle instance. 

The number of vertices in the DAG $\H_n$ is $O(n^3)$. The degree of the cubes being at most $3$, exploring the connective component of the graph requires at most $O(n^3)$ operations, which makes an algorithm of cubic complexity for solving an instance of the puzzle in $\triangle _n$. It proves Theorem \ref{maincalissons} and the algorithm used is a connective component exploration, i.e. the most elementary algorithm in the graph algorithmic arsenal.

It shows that if an instance of the calissons puzzle admits a solution, then there exists a DAG cut/stepped surface of maximum height. By reversing the roles of $\back _n$ and $\front _n$, or simply by symmetry, there also exists a minimal solution. All solutions of the puzzle instance lie between these two extreme solutions/surfaces. The algorithm computing the connective component of $\back _n$ in the reverse graph of $H_n$ completed by the unbreakable edges $\Pi(X)$ is called the \textit{advancing surface algorithm}.

\subsection{With a Paper, a Pencil and a Rubber}

We explain now how to execute the advancing surface algorithm with a sheet of paper, a pencil and an rubber. The first remark is that our perception implements more easily the additive algorithm of the advancing surface than the subtracting algorithm that we have by using $\H_n$ and starting from $\front$.

\begin{figure}[ht]
  \begin{center}
		\includegraphics[width=0.95\textwidth]{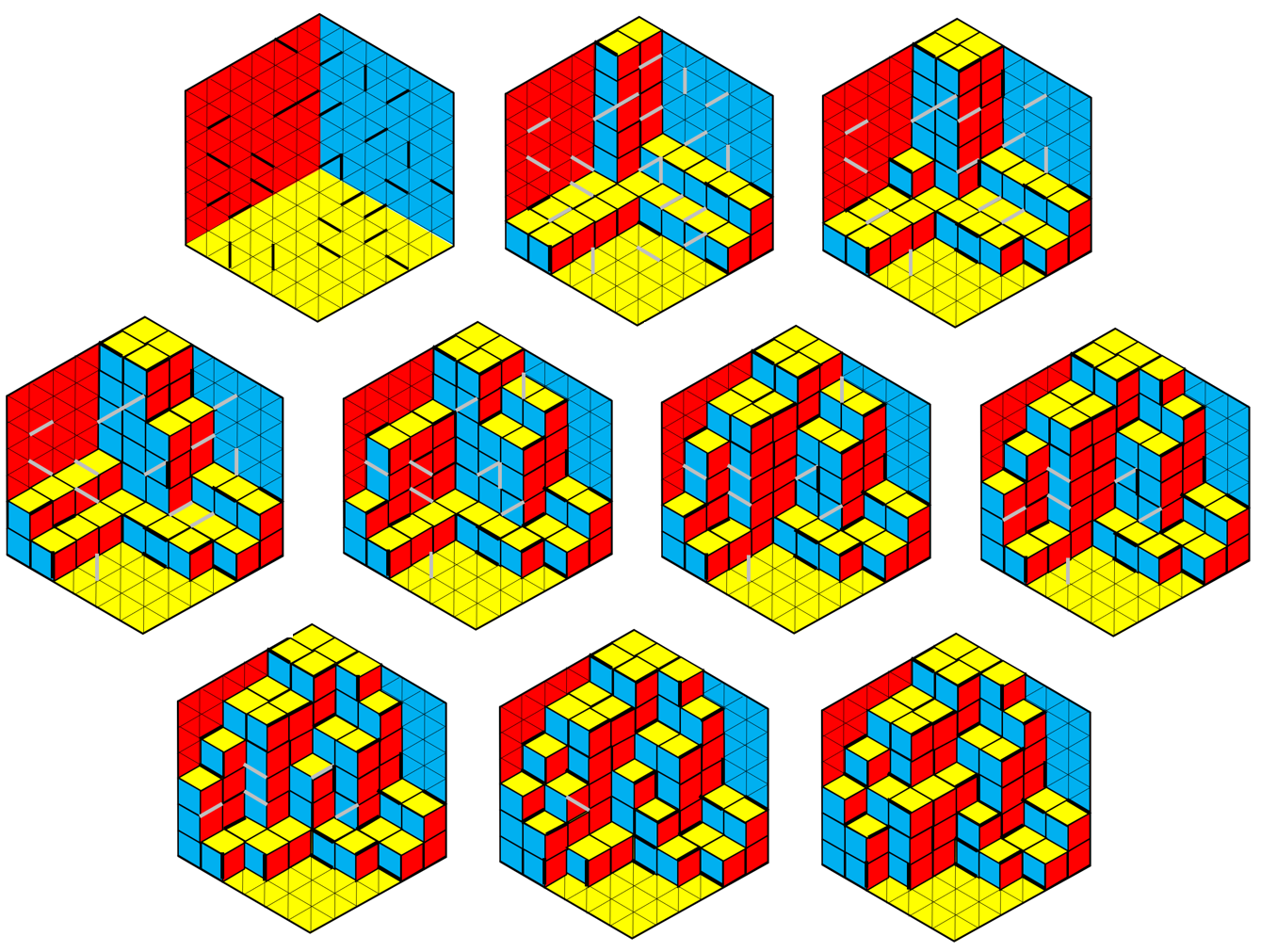}
	\end{center}
	\caption{\label{surfacequiavance} \textbf{Solving an instance of the calissons puzzle with the advancing surface algorithm.}}
\end{figure}

The strategy is illustrated in Fig.~\ref{surfacequiavance}.
A current stepped surface is initialized with the surface separating $\back _n$ from $\square _n$. The set of $\square _n$ cubes behind the surface is empty. To satisfy one of  edges $e\in X$, a cube of $\square _n$ must be added, along with all the cubes below it in the  DAG $\H_n$, i.e. backwards in the $(1,1,1)$ direction. With each addition, it must be ensured that the non-overlap and saliency constraints of the treated edge cannot be satisfied by adding a cube further back. Adding this cube may violate a previously satisfied constraint, but it is necessary. 
We therefore perform the operation of adding a cube and the cubes further back. On paper, we can even perform several operations in parallel on disjoint parts of the tiling. 
And so on until all the non-overlap and saliency constraints are satisfied, or until a cube of $\front _n$ is added, in which case the instance admits no solution.

\section{Solving the Extended Calissons Puzzle in Arbitrary Regions}\label{wwwalll}

The problem that we are now considering is more general. We want to tile a region $R\subset \triangle ^2$ with calissons. Our main assumption is that
$R$ is simply connected. The region $R$ is not necessarily bounded (we can have $R=\triangle ^2$). If it is bounded, its boundary is denoted $\partial R$ and we denote $\partial ^1 R$ its edges and $\partial ^0 R$ its vertices.  We admit the boundary to pass through the same vertices or edges several times but without imposing  the saliency constraint on the common edges. On the other hand, we exclude regions for which the set of triangles in $R$ is not connected according to edge adjacency (this convenient assumption does not reduce the generality of the framework since in this case, we can study the calissons puzzles independently in each connective component). An example of a finite region within the scope of this study is shown in Fig.\ref{duplication}. 

\begin{figure}[ht]
  \begin{center}
		\includegraphics[width=0.75\textwidth]{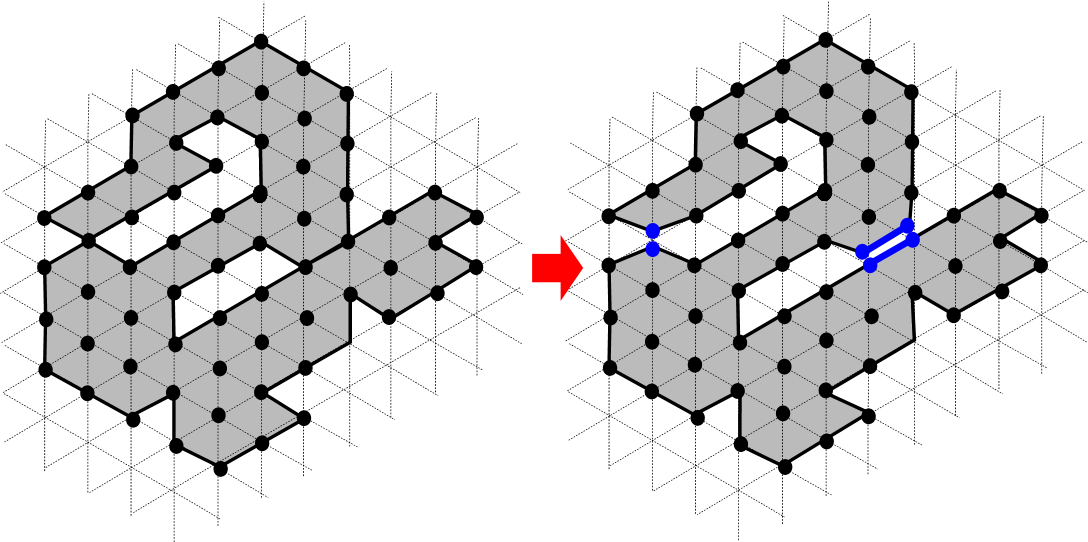}
	\end{center}
	\caption{\label{duplication} \textbf{A region $R$} in our scope. This region bounded by the path $\partial R$ is simply connected. In this case, the common vertices and edges of $\partial R$ are duplicated and are not subject to saliency constraints. The new sets of vertices and edges of $R$ are denoted $\triangle ^0 _R$ and $\triangle ^1 _R$.}
\end{figure}

To tile $R$, we impose the non-overlap condition (i) to obtain a tiling and possibly the saliency condition (ii). If we take into account the saliency conditions, the set of unbreakable edges is $\Pi(X)=\Pi _{(i)}(X) \cup \Pi _{(ii)} (X)$ while if we remove it, we have just $\Pi(X)=\Pi _{(i)}(X)$.
An instance of the extended calissons puzzle \texttt{Calissons$(X,R)$} is solved using different methods if the region $R$ is finite or not.

\subsection{The advancing surface algorithm}

To solve an instance of the calissons puzzle \texttt{Calissons$(X,R)$} with a finite and simply connected region $R$, we generalize the method of the advancing (backward of forward, as the case may be) surface presented to solve the problem in the hexagon. The main difference lies in the addition of a preliminary initialization step of the two sets $\back$ and $\front$. The algorithm is as follows:
 
\begin{enumerate}
    \item  Execute two times Thurston's algorithm to compute respectively the minimum and maximum tilings $P_{min}$ and $P_{max}$ of $R$. 
    Then we fix a pair of cubes $(x,y,z)+C, (x+1,y+1,z+1)+C$ whose projection $\varphi (x,y,z)$ is on the edge of $R$ and which we want to separate. The two tilings $P_{min}$ and $P_ {max}$ respectively define a minimal and maximal DAG cut of the set of cubes whose projection is in $R$ and separating $(x,y,z)+C, (x+1,y+1,z+1)+C$. We denote $\back _R$ the set of cubes below the minimum DAG cut and $\front _R$ the set of cubes above the maximum DAG cut.
    \item The two sets $\back _R$ and $\front _R$ now play the same role as  $\back _n$ and $\front _n$ in solving the initial puzzle \texttt{Calissons$(R,\triangle _n)$}. Let $\square ^3 _R$ be the set of cubes between $\back _R$ and $\front _R$. Finally, we define the DAG $\H _R$ induced by $\H$ on the set of vertices $\back _R \cup \square ^3 _R \cup \front _R$. The calisson tilings of the region $R$ are the $\varphi$ projections of the DAG cuts of $\H _R$ separating $\back _R$ from $\front _R$. To have a solution of an instance of \texttt{Calissons$(X,R)$}, the DAG cut must not cut any unbreakable edge. 
    So the algorithm simply computes the connective component of $\front _R$ in the graph $\H _R$ enriched with the set $\Pi$ of unbreakable edges. 
\end{enumerate}

In other words, the backward/forward surface algorithm agglomerates Thurston's algorithm to initialize $\back _R$ and $\front _R$ (computation time in $O(|\partial R |^2)$ ) with a connective component exploration in a graph of size $O(|\partial R |^3)$. The complexity of the algorithm is therefore $O(|\partial R |^3)$, which proves Theorem \ref{wall}.

\subsection{Extending Thurston's theorems and algorithm}

In the case of an instance \texttt{Calissons$(X,R)$} for an infinite region $R$, we can no more apply Thurston's algorithm or the advancing surface algorithm. The next results involve successive reductions of the instance \texttt{Calissons$(X,R)$} to three path problems in a graph.

\textbf{Notations.}
The region $R$ is bounded by $\partial R$.
Some vertices of $\partial ^0 _R$ and edges of $\partial ^1 _ R$ may appear several times (at most three) on the boundary of $R$. 
These vertices and edges are duplicated and attached to the various triangles and calissons to which they are connected (Fig.\ref{duplication}).
The part of the triangular grid covering $R$ and slightly modified by the duplications is denoted by 
$\triangle _R$ with $\triangle ^0 _R$, $\triangle _R ^1$ and $\triangle _R ^2$ as its set of vertices, edges and triangles.

We introduce the set $\overline \square ^3 _R$ of cubes $(x,y,z)+C$ whose
projection $\varphi(x,y,z)$ is a vertex of $\triangle _R ^0$. These are stacks of cubes in the  direction $(1,1,1)$. The overlined notation $\overline \square$ refers to the fact that there is no longer a stacking boundary. The heights of the cubes in a stack range from $-\infty$ to $+\infty$. As some of the vertices of $\triangle _R$ have been duplicated, so have the stacks of cubes that project onto them, and although we no longer mention it, most of the sets and relations presented in the following must take it into account. 

The set of cubes $\overline \square ^3 _R$ is completed by several sets of edges.

We start with the structural directed graph $\H_R=(\overline \square ^3 _R , \wedge _R)$ induced by the whole DAG $\H=(\square ^3 , \wedge)$ on the subset of cubes $\overline \square ^3 _R$.
This  graph denoted $\H_R=(\overline \square ^3 _R , \wedge _R)$ is a DAG. 
Note that the difference in height between the origin and destination cubes of any edge is $+1$. 
As it stands, a DAG cut of  $\H_R$ is not a calisson tiling of the region $R$, since the edges of the boundary $\partial R$ can be overlapped. 

To take into account the constraints of the calissons puzzle, we need to complete the graph $\H_R$ with the unbreakable edges that guarantee satisfaction of the constraints linked to the edge of $R$ and to $X$. 
According to the lemmas \ref{constraintz}, \ref{constraintx} and \ref{constrainty}, we have two types of unbreakable edges, non-overlapping and saliency edges, but if we also incorporate the non-overlapping edges of the edge of $R$, we have three classes of unbreakable edges:

\begin{enumerate}
    \item For an edge $e\in \partial R$, the set $\Pi_{(i)}(e)$ contains the unbreakable (two-way)  edges of non-overlapping of $e$. As rising edges are already considered in $\wedge _R$, we focus on the descending edges of $\Pi_{(i)}(\partial R)$ with vertices in $\overline \square ^3 _R$. Their set is denoted $\vee _R$. They descend by one unit.
    \item In the same way as $\vee _R$, we have the unbreakable edges 
     of $X$. As their upward direction is already taken into account in $\wedge _R$, we note 
    $\vee _X$ the set of descending edges for the non-overlapping constraints induced by $X$ in the downward direction. Their height difference is $-1$. 
    \item Finally, we have the unbreakable edges expressing the saliency constraints induced by the edges in $X$. They are two-way and have not yet been taken into account. Their height difference is $0$. Their set is denoted $\lessgtr _X$. 
\end{enumerate}

\begin{figure}[ht]
  \begin{center}
		\includegraphics[width=0.95\textwidth]{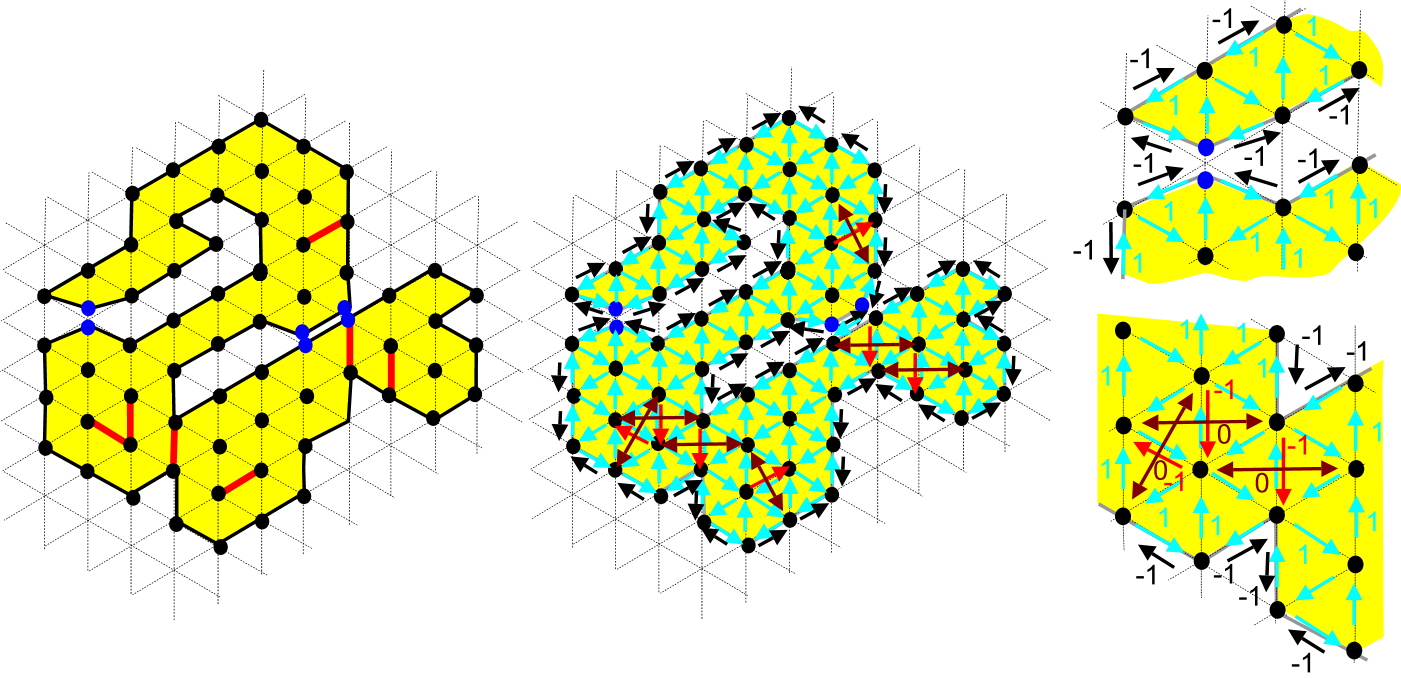}
	\end{center}
	\caption{\label{graphplan} \textbf{The projected graph $\varphi (\overline \square ^3 _R , \wedge _R \cup \vee _R \cup \vee _X \cup \lessgtr _X)$}. On the left, an instance of the calissons puzzle \texttt{Calissons$(X,R)$} with a region $R$ whose vertices and edges have been duplicated. In the center, the projected graph with  $\varphi(\wedge _R)$ edges with weight $+1$ (light blue), $\varphi(\vee _R)$ edges with weight $-1$ (black), $\varphi(\vee _X)$ edges with weight $-1$ (red) and $\varphi(\lessgtr _X)$ two-way saliency edges with weight $0$ (brown). On the right, a zoom on two zones.  }
\end{figure}

Finally, we introduce the projection
of the graph $(\overline \square ^3 _R , \wedge _R \cup \vee _R \cup \vee _X \cup \lessgtr _X)$ by $\varphi$ (see Fig.~\ref{graphplan}). By definition, the cubes of $\overline \square ^3 _R$ project onto the vertices of the region $R$, i.e. into $\triangle ^0 _R$. The edges of $\vee _R$ project onto the edges of $R$. The edges of $\vee _X$ project onto $X$. The edges of $\lessgtr _X$ project onto the diagonals of the calissons overlapping the edges of $X$ that are not in $X$. To compensate for the $2$ dimensions of this graph, each  edge $\varphi (e)$ projected from an $e$ edge is weighted by the height difference between its destination cube and its source cube. The weight of the edges in $\varphi(\wedge _R)$ is $+1$, the weight of the edges in $\varphi  (\vee _R \cup \vee _X)$ is $-1$ while the edges $\varphi  (\lessgtr _X)$ have a null weight. This projected weighted graph is denoted $\varphi (\overline \square ^3 _R , \wedge _R \cup \vee _R \cup \vee _X \cup \lessgtr _X)$ (Fig.~\ref{graphplan}).

\textbf{What do we get ?} 
The problems of tilability and of calissons puzzles in the region $R$ are expressed via the DAG $\H_R=(\overline \square ^3 _R , \wedge _R)$ and the sets of edges $\vee _R$, $\vee _X$ and $\lessgtr _X$.

A stepped surface of $R$ is then defined as a DAG cut of  $\H_R=(\overline \square ^3 _R , \wedge _R)$ which does not cut any edge of $\vee _R$.
Since the region $R$ is assumed to be simply connected, we still have a bijection between the tilings of $R$ and stepped surfaces. 

A stepped surface of $R$ solving a calissons puzzle \texttt{Calissons$(X,R)$} is a DAG cut of $\H_R=(\overline \square ^3 _R , \wedge _R)$  which does not cut any edge of $\vee _R \cup \vee _X \cup \lessgtr _X$.

For a finite  region $R$, we solve the problem by framing it by the minimal  and maximal  stepped surfaces $\back _R$ and $\front _R$. It reduces the problem to a DAG cut problem in a finite DAG and we solve it with a connective component exploration. 
For an infinite  region $R$, this is out of the question. Nevertheless, the problem can be rewritten in three different ways.

\begin{theorem}\label{monte}
The following four propositions are equivalent for a finite or non-finite, simply connected region $R$:
\begin{enumerate}
    \item The instance \texttt{Calissons$(X,R)$} admits a solution.
    \item The DAG $\H_R=(\overline \square ^3 _R , \wedge _R)$ admits a DAG cut which does not cut any edge of $\vee _R \cup \vee _X \cup \lessgtr _X$.
    \item The graph $(\overline \square ^3 _R , \wedge _R \cup \vee _R \cup \vee _X \cup \lessgtr _X)$ contains no path descending from a cube $(x, y, z) + C \in \overline \square ^3 _R$ to a cube $(x - k, y - k, z - k) + C \in \overline \square ^3 _R$ with $k > 0$.
    \item The weighted projected graph $\varphi (\overline \square ^3 _R , \wedge _R \cup \vee _R \cup \vee _X \cup \lessgtr _X)$ contains no absorbing cycle (Fig.~\ref{abs}).
\end{enumerate}
\end{theorem}

\begin{figure}[ht]
  \begin{center}		\includegraphics[width=0.45\textwidth]{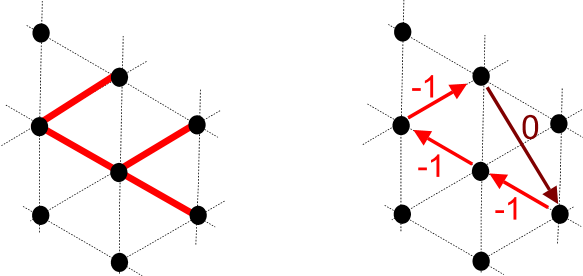}
	\end{center}
	\caption{\label{abs} \textbf{The instance \texttt{Calissons$(X,R)$} has no solution} if and only if the projected graph $\varphi(\overline \square ^3 _R , \wedge _R \cup \vee _R \cup \vee _X \cup \lessgtr _X)$ has an absorbing cycle. The above instance has no solution because of the saliency condition and we find an absorbing cycle passing through the edges of $\varphi(\lessgtr _X)$.}
\end{figure}

Thurston's results relate to the case where $X$ is empty. At the time, it was only a question of tilability.
The characterization of  surfaces which are tilable by calissons given in Theorem \ref{Thurston} is a corollary of the equivalence between propositions (1) and (4) of Theorem \ref{monte} in the case where $X$ is empty. 

Thurston's algorithm can also be generalized to a region $R$ with a non-empty edge set $X$.
We first explain why a distance computation algorithm in the projected graph $\varphi (\overline \square ^3 _R , \wedge _R \cup \vee _R \cup \vee _X \cup \lessgtr _X)$ allows us to solve the calissons puzzle and then show that Thurston's algorithm is a Dijkstra-like algorithm computing those distances when $X$ is empty.

The distance computation algorithm in the weighted projected graph $\varphi (\overline \square ^3 _R , \wedge _R \cup \vee _R \cup \vee _X \cup \lessgtr _X)$ starts by choosing any source vertex $s=\varphi(x_0,y_0,z_0)$ in $\triangle _R ^0$. We assume that the graph does not contain any absorbing cycles. Then the algorithm computes the distances $d(s,\varphi(x,y,z))$ from $s$ to any vertex of the $\triangle _R ^0$. As the edges weights correspond to the height differences between the cubes of $\overline \square _R ^3$, each distance $d(s,\varphi(x,y,z))$ is the height difference $h(x,y,z)-h(x_0,y_0,z_0)$ where $h(x_0,y_0,z_0)$ is the height of a fixed source cube $C_0$ above the source vertex and where  $h(x,y,z)$ is the height of the lowest cube of the stack above  $\varphi(x,y,z)$ belonging to the connective component of the source cube $C_0$ in $(\overline \square ^3 _R , \wedge _R \cup \vee _R \cup \vee _X \cup \lessgtr _X)$. In other words, the distances are the heights of a lowest layer of a connective component of the graph $(\overline \square ^3 _R , \wedge _R \cup \vee _R \cup \vee _X \cup \lessgtr _X)$. It provides a DAG cut or stepped surface solution of the instance \texttt{Calissons$(X,R)$}.

We now show that the heights computed by Thurston's algorithm in the case where $X$ is empty and by fixing the height of $s\in \partial R ^0$ at $0$ are exactly the distances
$d(s,\varphi(x,y,z))$. Thurston's algorithm starts by computing the heights of the boundary vertices by considering only the boundary edges. The computed heights might be larger than the distances $d(s,\varphi(x,y,z))$ since only the boundary edges are used for its computation, but if the interior edges provides a shortcut, there is an absorbing cycle and it is the case without solution. If there is no absorbing cycle, the heights computed along the boundary are the exact distances  $d(s,\varphi(x,y,z))$. The decimation routine of Thurston's algorithm is identical to Dijkstra's algorithm for computing the distances $d(s,\varphi(x,y,z))$. It considers the vertex $v$ of smallest computed distance to the source, updates the distances from the source to the neighbors of $v$ and never goes back to $v$. The guarantee that we do not have to revisit $v$ does not hold with negative weights which makes Dijkstra and Thurston's algorithm inefficient in this case. Then if we want to generalize Thurston's algorithm with non empty sets $X$, the extended algorithm has to deal with edges of negative weights. It requires to use Bellman-Ford's algorithm instead of Dijkstra's strategy. As conclusion, Thurston's algorithm can be generalized by the computation of the distances from a chosen source in the weighted projected graph $\varphi (\overline \square ^3 _R , \wedge _R \cup \vee _R \cup \vee _X \cup \lessgtr _X)$ with Bellman-Ford's algorithm \cite{Bellman}. Either the algorithm finds an absorbing cycle and there is no solution, or it provides the distances of each vertex and it remains to connect by segments the adjacent vertices whose distances to $s$ differ by $1$. The generalized Thurston's algorithm is illustrated Fig.~\ref{end}.

In the case of a finite region $R$, the time complexity of the distances computation by Bellman-Ford algorithm is  $O(|V||E|)$ namely $O(|\partial R|^4)$ because we have $O(|\partial R|^2)$ vertices and $O(|\partial R|^2)$ edges. It follows that this generalized version of Thurston's algorithm does not improve the cubic complexity  of the surface advancing algorithm going from $\back _R$ to $\front _R$.

\begin{figure}[ht]
  \begin{center}		\includegraphics[width=0.85\textwidth]{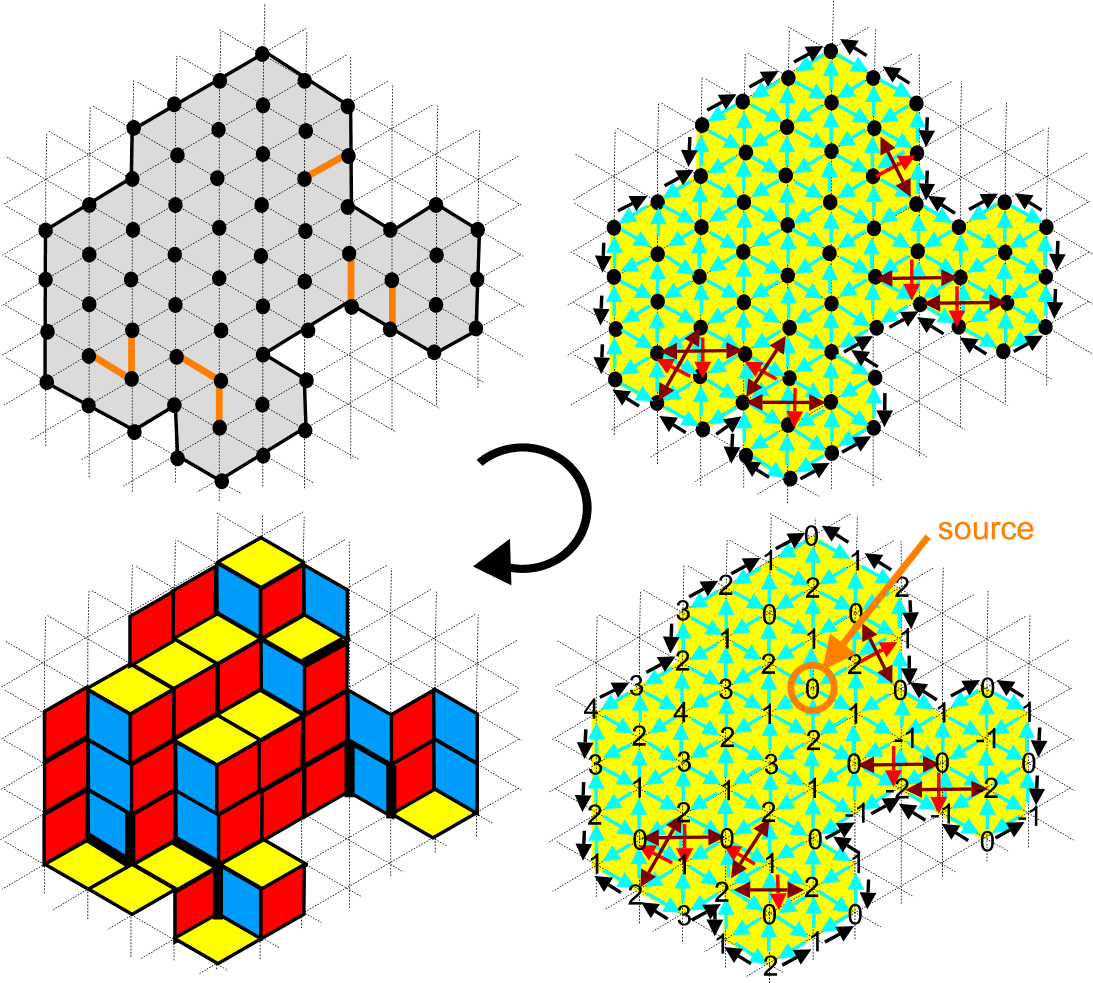}
	\end{center}
	\caption{\label{end} \textbf{The resolution of an instance \texttt{Calissons$(X,R)$}} by computing the distances from any source (in orange) in the projected graph $\varphi (\overline \square ^3 _R , \wedge _R \cup \vee _R \cup \vee _X \cup \lessgtr _X)$ (the weights of the blue, red and brown edges are respectively $+1$, $-1$ and $0$). }
\end{figure}


\subsection{Proof of Theorem \ref{wallinfini}}

The most useful proposition of Theorem \ref{monte} for solving an instance of \texttt{Calissons$(X,R)$} with an infinite region $R$ is proposition (4), but the graph $\varphi (\overline \square ^3 _R , \wedge _R \cup \vee _R \cup \vee _X \cup \lessgtr _X)$ still has an infinite number of vertices. The final step is to reduce it. To this end, we distinguish two classes of vertices.
We denote $X^0$ the vertices of the edges of $X$ and $\partial R ^0$ the vertices of the edges of $R$. 

\begin{itemize}
\item The \textit{regular} vertices  of the projected graph $\varphi (\overline \square ^3 _R , \wedge _R \cup \vee _R \cup \vee _X \cup \lessgtr _X)$ are the vertices of $\triangle ^0 _R$ adjacent only to edges of weight $+1$.
    \item The \textit{critical} vertices are the vertices of $\triangle ^0 _R$ adjacent to at least one edge of weight $0$ or $-1$. The set of critical vertices is $\partial R ^0 \cup X^0$. If $X$ is finite, there is a finite number of critical vertices.
\end{itemize}  

We reduce the graph $\varphi (\overline \square ^3 _R , \wedge _R \cup \vee _R \cup \vee _X \cup \lessgtr _X)$ to a graph denoted $\Gamma (R,X)$ with vertex set $\partial R ^0 \cup X^0$. In other words, all regular vertices are removed from the projected graph. This pruning is accompanied by the addition of new edges to make the directed graph $\Gamma (R,X)$ complete. 
Deleting regular vertices destroys many paths linking critical vertices but consisting of edges of weight $+1$. This is why we complete the edges of the graph $\Gamma (R,X)$. If there are no edges of weight $0$ or $-1$ going from $a$ to $b$, we add one of weight equal to the distance from $a$ to $b$ in the subgraph of the ascendant edges i.e. $\varphi (\overline \square ^3 _R , \wedge _R)$. These new edges compensate for the deleted vertices.
We then have the following equivalence.

\begin{lemma}\label{reducgraphe}
The graph $\varphi (\overline \square ^3 _R , \wedge _R \cup \vee _R \cup \vee _X \cup \lessgtr _X)$ contains an absorbing cycle if and only if the reduced graph $\Gamma (R,X)$ contains an absorbing cycle.
\end{lemma}
\begin{proof}
    If the graph $\varphi (\overline \square ^3 _R , \wedge _R \cup \vee _R \cup \vee _X \cup \lessgtr _X)$ contains an absorbing cycle, the cycle necessarily contains a critical vertex $a$. We can reconstruct the absorbing cycle of $\varphi (\overline \square ^3 _R , \wedge _R \cup \vee _R \cup \vee _X \cup \lessgtr _X)$ in $\Gamma (R,X)$ by following the critical vertices of the path and using the shortcuts of the new weighted edges when the path passes through regular vertices.

    Conversely, an absorbing cycle in the reduced graph $\Gamma (R,X)$ provides an absorbing cycle in the graph $\varphi (\overline \square ^3 _R , \wedge _R \cup \vee _R \cup \vee _X \cup \lessgtr _X)$ by following the shortest paths in $\Gamma (R,X)$ from a critical vertex to a critical vertex.
\end{proof}

The lemma \ref{reducgraphe}  makes instances of the puzle instances \texttt{Calissons$(X,R)$} decidable for certain unbounded  regions. 
The key point is the computation of the graph $\Gamma (R,X)$ which requires the computation of distances in $\varphi (\overline \square ^3 _R , \wedge _R)$. 

If we choose the region $R$ consisting of the entire triangular grid $\triangle$, the distances in $\varphi (\overline \square ^3 _R , \wedge _R)$ are computed in constant time.
The distance from $\varphi(x,y,z)$ to $\varphi(x',y',z')$ in $\varphi (\overline \square ^3 _R , \wedge _R)$ is equal to 
$(x'-x)+(y'-y)+(z'-z)-3 \mathrm{min}\{ (x'-x), (y'-y), (z'-z) \}$.


\begin{figure}[ht]
  \begin{center}
		\includegraphics[width=0.75\textwidth]{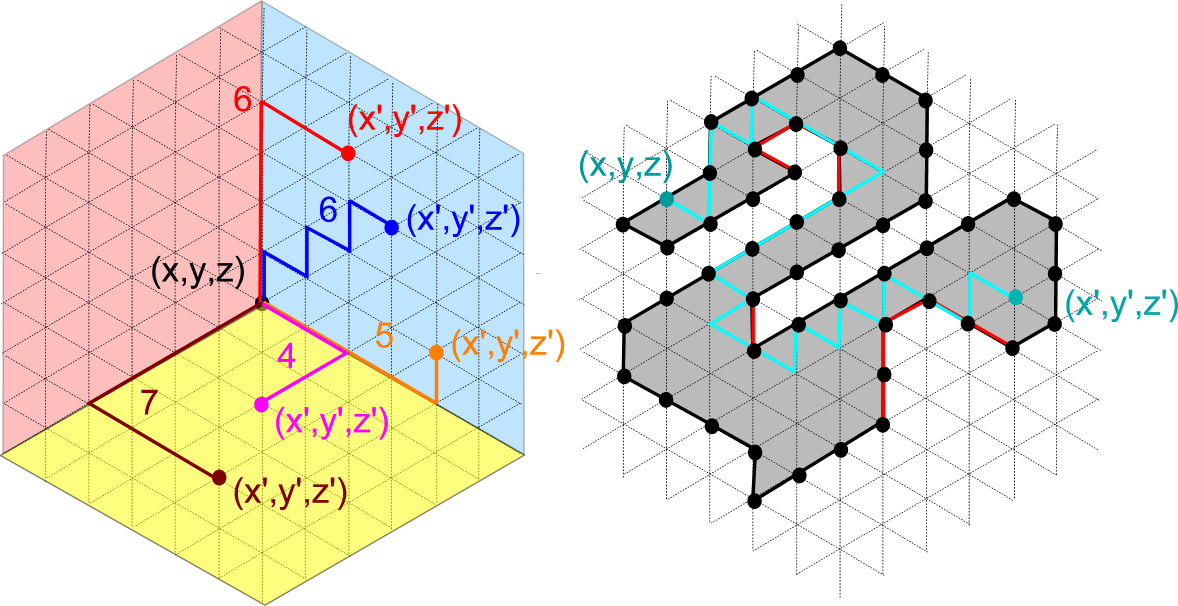}
	\end{center}
	\caption{\label{distances} \textbf{For the region $R=\triangle$, the distances in the graph $\varphi (\overline \square ^3 _R , \wedge _R)$} are computed in constant time by a formula (left). For convoluted regions, this can become more complicated}.
\end{figure}

If the region $R$ is the entire  grid $\triangle$, the vertices of the graph $\Gamma (\triangle,X)$ are the vertices of $X^0$. Their number is $O(|X|)$. The graph has $O(|X|)$ vertices and $O(|X|^2)$ edges whose weights are computed in constant time. Creating the graph $\Gamma (\triangle,X)$ takes $O(|X|^2)$ operations. 
Then, the search for an absorbing cycle in $\Gamma (\triangle,X)$ can be solved by the Bellman-Ford algorithm from any vertex of the graph \cite{Ford56,Bellman}. Its complexity is the product of the number of edges and vertices. We can therefore determine the existence of an absorbing cycle in $\Gamma (\triangle,X)$ in $O(|X|^3)$ operations. Combining Lemma \ref{reducgraphe} with proposition (4) of Theorem \ref{monte}, the absence of an absorbing cycle in $\Gamma (\triangle,X)$ is equivalent to the existence of a solution to the instance \texttt{Calissons$(X,R)$}. This proves Theorem \ref{wallinfini}.


\subsection{Proof of Theorem \ref{monte}.} 

This proof can be written with different levels of detail.

\begin{proof}
We assume (1) and show (2). 
A set of heights can be defined by tilings. First, we identify a vertex $\varphi (x,y,z)$ of the tiling at height $0$. Then, by following the edges of the tiling, we can compute the heights of all the vertices in the tiling. The fact that the region $R$ has no holes guarantees the consistency of the heights (whatever the path taken to go from $\varphi (x,y,z)$ to $\varphi (x',y',z')$, the height obtained is identical, as one path can be deformed into another without changing the initial and final heights). Each vertex $\varphi (x',y',z')$ is then associated with the cube $(x'+y'+z')+C$ whose height $x'+y'+z'$ is the height computed with the tiling. We thus obtain a set $L$ of cubes such that $\varphi (B)=\triangle _R$. In $\H_R$, consider the DAG cut that separates the cubes strictly above $L$ from the cubes at $L$ and below. 
We have to prove now that this DAG cut does not cut an unbreakable edge in $\vee _R \cup \vee _X \cup \lessgtr _X$. 

Consider an edge $e$ of $\vee _R \cup \vee _X$ connecting $(x,y,z)+C$ to $(x-1,y,z)+C$. For a solution of the instance \texttt{Calissons$(X,R)$}, the edge $\varphi(e)$ is a tiling edge (not overlapped by a calisson). If the height 
 computed from the tiling of the highest cube $(x,y,z)+C$ of projection $\varphi (x,y,z)$ is denoted $h(x,y,z)$, then the height $h(x-1,y,z)$ of the highest cube $(x,y,z)+C$ of projection $\varphi (x-1,y,z)$ is $h(x-1,y,z)=h(x,y,z)-1$. This shows that if the origin $(x,y,z)+C$  of edge $e$ is under the DAG cut, then so is the end cube $(x-1,y,z)+C$ of the edge $e$.

Consider an edge $e$ of $\lessgtr _X$ connecting $(x,y,z)+C$ to $(x-1,y+1,z)+C$. For a solution of the instance \texttt{Calissons$(X,R)$}, the tiling has two calissons of different colors adjacent to $\varphi(e)$, making two calisson edges from $\varphi (x,y,z)$ to $\varphi (x-1,y+1,z)$ preserving the height. If the height 
 computed from the tiling of the highest cube $(x,y,z)+C$ of projection $\varphi (x,y,z)$ is denoted $h(x,y,z)$, then the height $h(x-1,y+1,z)$ of the highest cube $(x,y,z)+C$ of projection $\varphi (x-1,y+1,z)$ is $h(x-1,y+1,z)=h(x,y,z)$. This shows that if $(x,y,z)+C$ origin of edge $e$ is under the DAG cut, then so is the end cube $(x-1,y+1,z)+C$ of the edge $e$.

We now prove that (2) implies (3) by establishing that (2) and not (3) lead to a contradiction.
 The proof is based on the idea that if we have a cube $(x,y,z)+C$ above the cut, then a path from $(x,y,z)+C$ in the graph $(\overline \square ^3 _R , \wedge _R \cup \vee _R \cup \vee _X \cup \lessgtr _X)$ cannot be cut because it is made up of unbreakable edges and edges $a \rightarrow b$ edges with $a<b$ ($a$ cannot be above the DAG cut without $b$ being there too). In other words, if $(x,y,z)+C$ is above the DAG cut, all vertices related to it in $(\overline \square ^3 _R , \wedge _R \cup \vee _R \cup \vee _X \cup \lessgtr _X)$ are also above the DAG cut. The assumption not (3) means that there is a descending path 
 in $(\overline \square ^3 _R , \wedge _R \cup \vee _R \cup \vee _X \cup \lessgtr _X)$. Since the graph $(\overline \square ^3 _R , \wedge _R \cup \vee _R \cup \vee _X \cup \lessgtr _X)$ is invariant by translation of vector $(1,1,1)$, there is a path traversing $(\overline \square ^3 _R , \wedge _R \cup \vee _R \cup \vee _X \cup \lessgtr _X)$ from height $+\infty$ to $-\infty$. It implies that the connective component of any cube in $(\overline \square ^3 _R , \wedge _R \cup \vee _R \cup \vee _X \cup \lessgtr _X)$ entirely contains $\square ^0 _R$, which contradicts the fact that we have a DAG cut and leads to a contradiction.
 
We now prove that (3) implies (4). The proof simply consists in noticing
that the graph 
$(\overline \square ^3 _R , \wedge _R \cup \vee _R \cup \vee _X \cup \lessgtr _X)$ has a descending path if and only if the projected graph 
$\varphi(\overline \square ^3 _R , \wedge _R \cup \vee _R \cup \vee _X \cup \lessgtr _X))$ in which height differences are represented by weights, contains an absorbing cycle.

Finally, we show that (4) implies (1) by describing the computation of a tiling from the graph $\varphi(\overline \square ^3 _R , \wedge _R \cup \vee _R \cup \vee _X \cup \lessgtr _X))$. The process is the generalized Thurston algorithm illustrated in Fig.~\ref{end}. We choose a source vertex $\varphi(x,y,z)\in \triangle^0 _R$ and set its height to $0$. We compute the distances to this vertex in the weighted projected graph $\varphi(\overline \square ^3 _R , \wedge _R \cup \vee _R \cup \vee _X \cup \lessgtr _X)$. Adjacent vertices in the triangular grid $\triangle_R$ whose distance/height differs from $1$ are connected by an edge and those whose distance/height differs from $0$ or $2$ are not connected. 
The weights of the $\wedge _R \cup \vee _R \cup \vee _X \cup \lessgtr _X$ edges guarantee that the tiling respects the non-overlap and saliency constraints of the \texttt{Calissons$(X,R)$} instance.
\end{proof}

\subsection{Conclusion and Open Questions}

We have provided a general solution to the calissons puzzle problem (with or without saliency constraint) for a  region without holes.
This work revisits and extends the legacy of William Thurston with a computational tone. We have used the notions of DAG cuts and the associated algorithmic through two elementary graph algorithms, the exploration of a connective component and the calculation of distances with Bellman-Ford algorithm.
However, it remains at least two open questions:
\begin{itemize}
    \item For a region $R$ with (non tilable) holes,  the calisson tilings are  no more DAG cuts. They are closer from covering spaces of the region $R$ in $\H_R$. In this more complex setting, is the calissons puzzle still solvable in polynomial time?
    \item Can the calissons puzzle and the results that we have established be extended to domino tilings in a square grid (a framework in which the notion of height can also be used)?
\end{itemize}

\bibliography{lipics-v2021-sample-article}

\end{document}